\documentclass{article}

\usepackage{arxiv}

\usepackage[utf8]{inputenc} 
\usepackage[T1]{fontenc}    
\usepackage{hyperref}       
\usepackage{url}            
\usepackage{booktabs}       
\usepackage{amsfonts}       
\usepackage{amsthm}
\usepackage{nicefrac}       
\usepackage{microtype}      
\usepackage{lipsum}		
\usepackage{algorithm2e}
\usepackage{algorithmic}
\usepackage{amsfonts}
\usepackage{amsmath}
\usepackage{mathtools}
\usepackage{graphicx}
\usepackage{hyperref}

\newtheorem{proposition}{Proposition}
\newtheorem{lemma}{Lemma}
\newtheorem{theorem}{Theorem}
\newtheorem{corollary}{Corollary}

\newcommand{\REAL}{\ensuremath{\mathbb{R}}}

\newcommand{\COMPLEX}{\ensuremath{\mathbb{C}}}
\newcommand{\Exp}[1]{\mathbb{E}\left[#1\right]}
\newcommand{\Prob}[1]{\text{Prob}\left[#1\right]}

\title{Collaborative Broadcast in $\mathcal{O}(\log \log n)$ Rounds}
\date{August 30, 2019}	

\author{Christian Schindelhauer\\
 University of Freiburg,\\
Georges-Köhler-Allee 51,\\
79110 Freiburg im Breisgau, Germany\\
\texttt{schindel@tf.uni-freiburg.de}
\And
Aditya Oak\\
Technical University of Darmstadt, \\
Hochschulstra{\ss}e 10,\\
64289 Darmstadt, Germany\\
\texttt{oak@st.informatik.tu-darmstadt.de}
\And 
Thomas Janson\\
 University of Freiburg,\\
Georges-K{\"o}hler-Allee 51,\\
79110 Freiburg im Breisgau, Germany\\
\texttt{thomas@janson-online.de}
}
\begin{document}
\maketitle
\begin{abstract}
We consider the multihop broadcasting problem for $n$ nodes placed uniformly at random in a disk and investigate the number of hops required to transmit a signal from the central node to all other nodes under three communication models: Unit-Disk-Graph (UDG), Signal-to-Noise-Ratio (SNR), and the wave superposition model of multiple input/multiple output (MIMO).

In the MIMO model, informed nodes cooperate to produce a stronger superposed signal. We do not consider the problem of transmitting a full message nor do we consider interference with other messages. In each round, the informed senders try to deliver to other nodes the required signal strength such that the received signal can be distinguished from the noise.

We assume a sufficiently high node density $\rho= \Omega(\log n)$ in order to launch the broadcasting process. In the unit-disk graph model, broadcasting takes $\mathcal{O}(\sqrt{n/\rho})$ rounds. In the other models, we use an Expanding Disk Broadcasting Algorithm, where in a round only triggered nodes within a certain distance from the initiator node contribute to the broadcasting operation.

This algorithm achieves a broadcast in only $\mathcal{O}\left(\frac{\log n}{\log \rho}\right)$ rounds  in the SNR-model.
Adapted to the MIMO model, it  broadcasts within $\mathcal{O}(\log \log n - \log \log \rho)$ rounds. All bounds are asymptotically tight and hold with high probability, i.e. $1- n^{-\mathcal{O}(1)}$.
\end{abstract}
\begin{figure}
\centering
  \includegraphics[width=.7\textwidth]{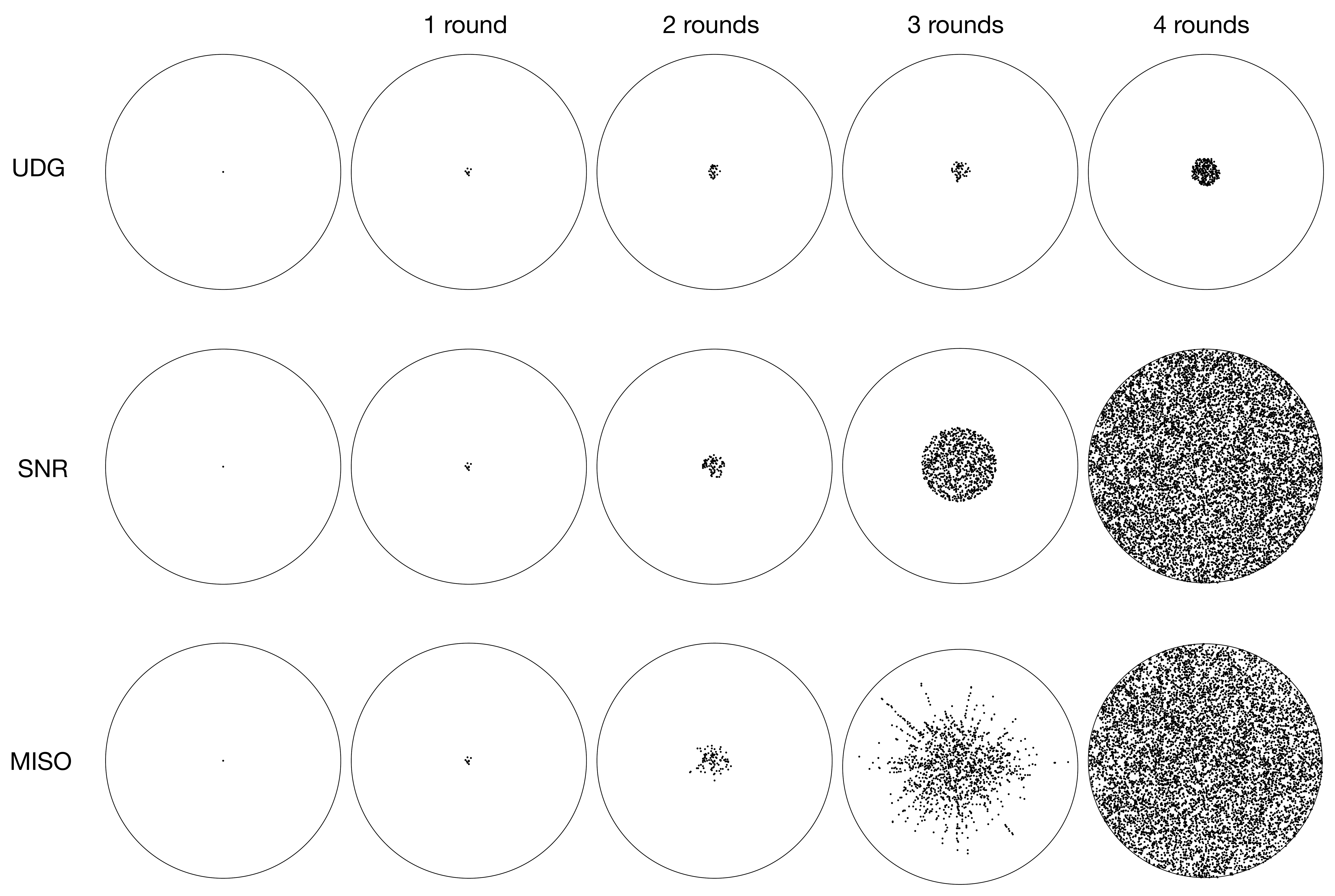}
  \caption{Four rounds of repeated collaborative broadcast in the UDG, SNR, and MIMO model for
  10,000 senders randomly distributed in a disk with radius 30 and wavelength $\lambda= 0.1$}
  \label{fig:teaser}
\end{figure}

%

\section{Introduction}
Understanding the limits of multi hop communications and broadcasting is important for the development of new technologies in the wireless communication sector. In the recent decades, ever more realistic models for communication have been considered. First, graph models have been used to describe the communication between wireless communication nodes, resulting in the Radio Broadcast model \cite{peleg2007time}. However, this model neglects the communication range, which has led to a geometric graph model, the Unit-Disk Graph (UDG) \cite{UnitDiskGraphs1990}, which we also consider here. It is based on the observation that there is a path loss of the sender energy with increasing distance between the sender and the receiver. In order to distinguish the signal from noise, the signal to noise energy ratio (SNR) has to be above certain threshold, which leads to the disk shaped model for radio coverage. 

However, if one carefully models the influence of the noise on the SNR, one sees that interference-free communication links inside a disk are still possible because of the polynomial nature of the energy path loss. This model is known as the SINR-model \cite{Goussevskaia07}. This, however, is still far from reality, where one sees a superposition of electromagnetic waves, which can be expressed by the addition of the complex Fourier coefficients. In this model, diversity gain and energy gain \cite{Tse_fundamentals_book} enable higher bandwidth and higher communication range. There is a trade-off between these two features and certain properties of the number of coordinated senders, receivers and the channel matrix have to be met. However, one sees that in practice this one-hop communication has already led to better networking solutions.

For our theoretical analysis, we concentrate on an open space model with no interfering communications. We want to find the theoretical limitations of a collaborative multi-hop broadcast. For this, we are interested in sending a carrier signal with no further modulated information. This signal is sent by the sender node positioned at the center of a disk in which all other nodes are randomly distributed.
Thus, in the first round the first sender activates some small number of neighboring nodes. Then, in every subsequent round, all of them try to extend the set of informed nodes as far as possible, who then join in the next round, until all nodes of the disk are informed (or the process cannot reach any further nodes). 

\section{Related Work}

Broadcasting algorithms have been widely optimized for speed, throughput, and energy consumption.  A lot of algorithms apply MAC (medium access control) protocols like TDMA (Time Division Multiple Access) \cite{UnitDiskGraphs1990,gandhi2008minimizing,Halldorsson:2018:LIS:3212734.3212766,UnitDiskSINRLotker2009}, CDMA (Code Division Multiple Access) \cite{5779066,Sirkeci-Mergen_First}, FDMA (Frequency Division Multiple Access) \cite{Sirkeci-Mergen_First} to increase spatial reuse. Physical models with high path loss exponent $\alpha > 2$ are beneficial here and increase the spatial reuse with only local interference. With spatial reuse, parallel  point-to-point communications are possible which either spread the same broadcast message in the network or pipeline multiple broadcast messages at the same time. The latter can achieve a constant broadcasting rate for path loss exponent $\alpha > 2$. Cooperative transmission with MISO (Multiple Input Single Output) or MIMO (Multiple Input Multiple Output) is applied to increase the transmission range and broadcast speed by a constant factor (where underlying MAC protocols still work). 

Here, we focus on broadcast speed and allow as many as possible nodes cooperate in transmitting the same broadcast message with MIMO. The obvious trade-off here is broadcast speed against broadcast rate, since pipelining and spatial reuse are limited.

Broadcasting has been first considered for a graph based model, where interference prevents communication and a choice has to be made which link should be used for propagation. Since we do not consider interference and allow the usage of all links, a simple flooding algorithm achieves the optimal bound of the diameter of the network. So, these works (see \cite{peleg2007time} for a survey) do not apply here. However, even if interference is considered  there is only a constant factor slow down in the Unit-Disk-Graph model \cite{gandhi2008minimizing}. Note that Unit-Disk-Graphs are connected, when the node density of the randomly placed nodes is large enough \cite{xue2004number}.

Launched by the seminal paper of \cite{Gupta00thecapacity}, the SNR (Signal to Noise Ratio) model has gained a lot of interest. Here, signals can be received if the energy of the sending nodes is a constant factor larger than the sum of noise energy and interference. This model leads to a smooth receiver area with near convexity properties \cite{AvienEmek12}.

If the energy of each sender is constrained, Lebhar et al.~\cite{UnitDiskSINRLotker2009}  show that the SNR-model does not give much improvement compared to the UDG-model. So, they incorporate the unit disk model into the SINR (Signal to Interference and Noise) model. The focus of their work is finding TDMA scheduling schemes to enhance the network capacity while the path-loss exponent in the SINR model is chosen with $\alpha > 2$ such that interferences have only local effects for unsynchronized transmitters.
%
In this context, the SNR model is used for each sender separately. So, the problem of broadcast mainly reduces to range assignment and scheduling problem, for which the number of rounds approaches the diameter~\cite{Halldorsson:2018:LIS:3212734.3212766}.

For the superposition model the problem of point-to-point communication has been considered mostly for beam-forming for senders (MISO/MIMO) or receivers (SIMO/MIMO). For MIMO (Multiple Input Multiple Output), most of the research is concerned with the energy gain and diversity gain, as well as the trade-off. For an excellent survey we refer to \cite{Tse_fundamentals_book}. Besides the approach, where sender antennas and receiver antennas are connected to one device and only a one hop communication is considered, a lot of work is dedicated to collaboration of independent senders and receivers, for which we now discuss some noteworthy contributions.

A transmission with cooperative beamforming requires phase synchronization of the collaborating transmitters to produce a beam and sharing the data to transmit.
Dong et al.~\cite{dpp2000} present for this a two phase scheme: in phase one, the message is spread among nodes in a disk in the plane around the node holding the original message. The open-loop and closed-loop approach
can be used to synchronize nodes to the destination or a known node position and time synchronization. In phase two, the synchronized nodes jointly transmit the message towards the destination.

In~\cite{freitas2012energyWSN} a three phase scheme is presented. In order to save energy for a Wireless Sensor Network, in the first phase, a sensor sends its message via SIMO to a group of nearby nodes. In the second phase the nodes use MIMO beamforming to another group of nodes nearby of the receiver and in the final phase the last group of nodes sends the message via MISO to the recipient.

For the MIMO model in \cite{NGS09_Linear_Capacity_Beamforming} and
\cite{ozgur2010linearCapacity} the authors give a recursive construction, which provides a capacity of $n$ for $n$ senders using MIMO communication using its diversity gain. Yet, in \cite{franceschetti2009capacity}, an upper bound of $\sqrt{n}$ for such a diversity gain has been proved. These seemingly contradicting statements have been addressed
in  \cite{ozgur2013spatial}, where they address the question whether distributed MIMO provides significant capacity gain over traditional multi-hop in large ad hoc networks with $n$ source-destination pairs randomly distributed over an area A. It turns out that the capacity depends on the ratio $\sqrt{A}/\lambda$, which describes the spatial degree of freedom. If it is larger than $n$ it allows $n$ degrees of freedom \cite{ozgur2010linearCapacity}, if it is less than $\sqrt{n}$ the bound of \cite{franceschetti2009capacity} holds. For all regimes optimal constructions are provided in these papers.
While in \cite{ozgur2010linearCapacity} path loss exponents $\alpha\in (2,3]$ are considered, for $\alpha>3$ the regularity of the node placement must be taken into account \cite{NGS09_Linear_Capacity_Beamforming}.

While this research is largely concerned with the diversity gain, we study the physical limitations of the energy gain in MIMO. In 
\cite{oyman2007power,oyman2011cooperative}, a method is presented  to amplify the signal by using spatially distributed nodes. They explore the trade-off between energy efficiency and spectral efficiency with respect to network size.
 In \cite{mlo13_telescopic_beamforming}, a distributed algorithm is presented in which rectangular collaborative clusters of increasing size are used to produce stronger signal beams. 
 
 Janson et al.~\cite{JS14_Beamforming_LogLog_TR} analyze the asymptotic behavior of the rounds for a unicast in great detail and prove an upper and lower bound of $\Theta(\log \log n)$ rounds. If the nodes are placed on the line it takes an exponential number of rounds \cite{JS13_Beamforming_Line}. The generalization of these observations for different path loss models can be found in \cite{diss-janson-2015}. In  \cite{6962163} it is shown that the sum of all cooperating sender power can be reduced to the order of one sender, while maintaining a logarithmic number of rounds to send a message over an $n$ hop distance.

A practical approach already uses this technology. Glossy \cite{5779066} is a network architecture for time synchronization and broadcast including a network protocol for flooding, integration in network protocols of the application, and implementation in real-world sensor nodes. If multiple nodes transmit the same packet in a local area, the same symbol of the different transmitters will overlap at a receiver without inter-symbol interference if the synchronization is sufficient. The superposed signals of the same message have random phase shifts and in the expectation add up constructively. Faraway, out of sync, transmitters produce noise-like interference the influence of which is alleviated at the receiver via pseudo-noise codes. While a high node density increases interference in common network protocols, a higher density is beneficial here and increases the transmission range and reduces the number of broadcasting rounds. 

Glossy is the underlying technology for the so-called  Low-Power Wireless Bus~\cite{Ferrari:2012:LWB:2426656.2426658}, where this multi-hop broadcast allows to flood the network with a broadcasting message. The energy efficiency was further improved in Zippy \cite{sutton2015zippy}, which is an on-demand flooding technique providing robust wake-up in the network. Unlike Glossy, Zippy uses an asynchronous wake-up flooding. In \cite{kumberg2017exploiting}  the problem of Rayleigh fading for synchronized identical signals is addressed by producing a low frequency wake-up signal, which results from the beat frequency of closely chosen frequencies. This allows the usage of a passive receiver technology.

Sirkeci-Mergen et al.~\cite{Sirkeci-Mergen_First} propose a multistage cooperative broadcast algorithm similar to our work. Their nodes are also uniformly distributed in a disk. A continuum approximation is used to approximate the behavior of the disk with high node density. A minimum SNR threshold is assumed for successful reception of the message. Their algorithm works in stages, in the first stage, the node at the center of the disk transmits the message. All nodes which receive this message are considered as level one. In the next stage, level one nodes re-transmit the message, in this way set of informed nodes keeps growing in radially outward direction. Nodes belonging to same levels form concentric rings. Source node emits single block of data. 

A similar problem and a similar algorithm has been considered in~\cite{sirkeci2010broadcast}. Sirkeci-Mergen et al.~consider source node transmitting a continuous message signal. Initially source node which is at the center of the disk, transmits the message signal. In the next round, level one nodes, i.e. the set of nodes that received the message in the previous round, transmit the message signal which is received by next level and the source node does not transmit message. In the following round, the source transmits the next message block. In this way, levels send and receive the message block in alternate rounds. In our work, we consider that in each round, all informed nodes send a single message cooperatively and we prove bounds on the number of rounds needed.

Jeon et al.~\cite{jeon2007two} also consider a system model similar to our work. They use two phase opportunistic broadcasting to achieve linear increase in propagation distance. In phase one, nodes inside a disk of specific radius broadcast message with different random phases while in phase two, a node broadcasts the message to its neighboring nodes. These phases are performed repeatedly to broadcast the message. Improving on this work we obtain better bounds by coordinating the phase of the nodes, while we
consider only the path loss factor of $\alpha=2$.

To our knowledge, no research so far has evaluated the asymptotic number of rounds to cover the disk using cooperative broadcast using MIMO, which is the main focus of this work. While \cite{5779066,Ferrari:2012:LWB:2426656.2426658,kumberg2017exploiting,sutton2015zippy} use only simulation and \cite{jeon2007two,sirkeci2010broadcast,Sirkeci-Mergen_First} prove all their statements only for the expectation in the continuum limit, i.e. when the number of nodes approaches infinity. Our results are to our knowledge the first asymptotic results in MIMO that hold for a finite number of nodes $n$ with high probability, i.e. $1-n^{-\mathcal{O}(1)}$. 

\paragraph{Notations}
The $L_2$-norm is denoted by $|\!|p|\!|_2 = \sqrt{x^2  + y^2}$ for $p=(x,y) \in \REAL^2$.  
For representation of signal waves we use complex numbers  $\COMPLEX$ where  the imaginary number is $i=\sqrt{-1}$. For $z= a+b i$ the complex conjugate is $z^* = a-bi$, the absolute value $|z| = \sqrt{z \cdot z^*}=
\sqrt{a^2+b^2}$ and the real part is $\Re(z) = a = \frac{z+z^*}{2}$, the imaginary part is $\Im(z) = b = \frac{z-z^*}{2}$. The exponent for the base of the Euler number $e$ gives $e^{a+b i} = e^a (\cos b + i \sin b)$.
Note that $z = |z| \cdot e^{i \arg(z)}$, where $\arg(z) \in [0, 2\pi)$.

\section{The Models}
We assume $n$ nodes $v_1, \ldots, v_n \in \REAL^2$ uniformly distributed in a disk of radius $R$ centered at origin, where the additional node $v_0$ resides. The density is denoted by $\rho= n/(\pi R^2)$. Each node knows the disk radius $R$, its location and all nodes are perfectly synchronized, see Fig.~\ref{broadfig}.

\begin{figure}[ht]
\centerline{\includegraphics[width=1\linewidth]{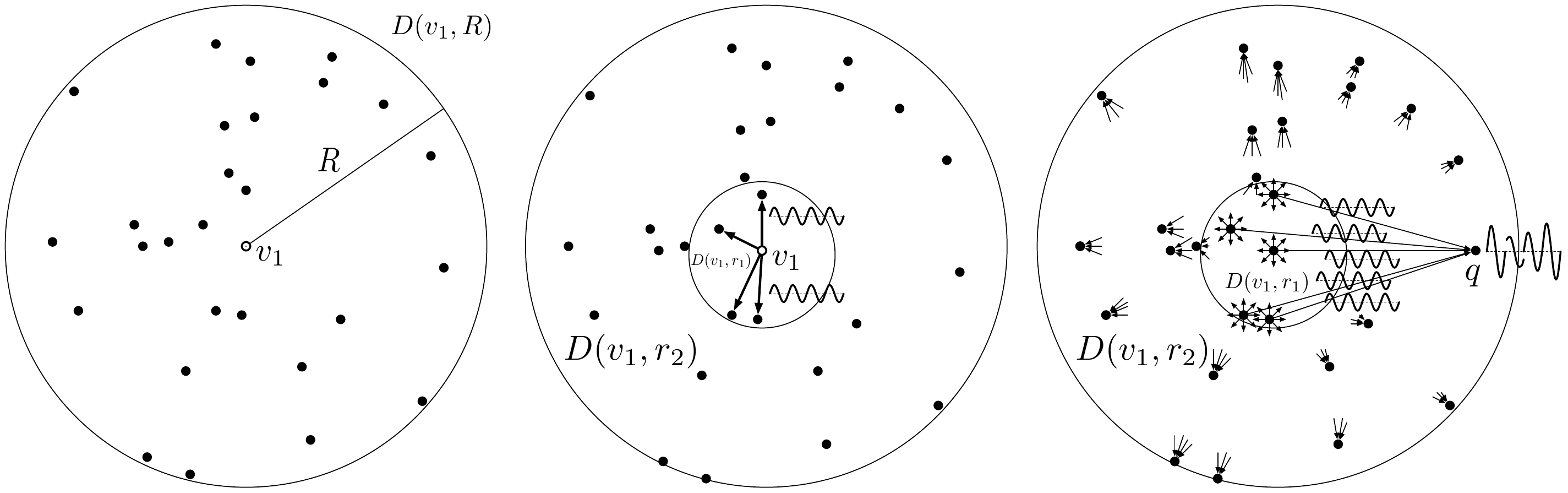}}
\caption{Randomly positioned nodes in a disk of radius $R$ with $v_1$ in the center. The first two phases of the considered collaborative Broadcast algorithms\label{broadfig}}
\end{figure}

We say that a node is triggered or informed, when it has received a signal carrying no further information. The objective is to send the broadcast signal from the center node $v_0$ to all other nodes where in each round the set of sending nodes is increased by the triggered nodes of the last round.  

We concentrate on broadcasting a pure sinusoidal signal and leave the problem of broadcasting a complete message to subsequent work.
The sinusoidal signal has wavelength $\lambda$ and we normalize the speed of light as $c=1$ by choosing proper units for time and space. In our theoretical framework we assume that every node knows its exact position in the plane, is synchronized (well enough in order to emit phase-coordinated signals) and is able to precisely emit the signal at a given point in time with a certain phase shift and a fixed amplitude.

We consider three communication models in our analysis: Unit-Disk-Graph (UDG),
 the Signal-to-Noise Ratio (SNR), 
 and MISO/MIMO (Multiple Input---Single/Multiple Output) for coordinated senders. 
The difference between MIMO and MISO is whether we consider a single receiver or multiple receivers. Since, MIMO is the more general term we prefer this term throughout this paper.

\subsection{The MIMO Model}
The coordination of nodes refers here to synchronized signals allowing  a radiation pattern containing strong beams, i.e. a beamforming gain.
Many physical properties are covered in the {\bf Multiple Input/Multiple Output (MIMO)} model based on superposition of waves.
Every node can serve either as sender or as receiver. A node can demodulate a received signal $rx(t) \in \COMPLEX$ if the square of the length of the Fourier coefficient over an interval of $\delta\gg \lambda$ is larger than $\beta$, i.e.
\begin{eqnarray}
 z &=&  \frac1{\delta}\int\displaylimits_{t=t_0}^{t_0+\delta} \hbox{rx}(t) \  e^{-i 2 \pi t/ \lambda} \mathrm{d}t \ ,\nonumber \\
 |z|^2 / N_0 & \geq & \beta \ .\label{eqsnr}
 \end{eqnarray}
with imaginary number $i=\sqrt{-1}$ and $t$ denoting time. In this notation we normalize the energy with respect to the time period and assume $\delta$, $N_0$ and $\beta$ are constant. 
The bound~(\ref{eqsnr}) demands that the signal-to-noise energy ratio is large enough to allow a successful signal reception,  i.e. $\text{SNR} \ge \beta$ for signal power $|z|^2$ and additive white noise with power $N_0$ over time $\delta$.

Each sending node $j \in \{1,\ldots, n\}$ can start sending at a designated time $t_1$ and stops at $t_2$, described by the function
\begin{equation} s_j(t) = \begin{cases}
  a \cdot e^{i 2\pi (t-t_1)/ \lambda} \ , & t \in [t_1,t_2] \ ,\\
  0 \ , & \hbox{otherwise} \ , \nonumber
  \end{cases}\end{equation}
  where $a\in \COMPLEX$ may encode some signal information, e.g. via Quadrature Amplitude Modulation (QAM). Since we are only interested in transmitting a single signal we choose $a=1$ or
  $a= e^{i \varphi}$, when we use a phase shift $\varphi$.
  The total signal received at a node $q \in \REAL^2$ is modeled by
\begin{equation}
 \hbox{rx}(t) = \sum_{j=1}^n \frac{s_j(t- |\!|q - v_j|\!|_2)}{|\!|q-v_j|\!|_2}\ , \nonumber
 \end{equation} 
which models the free space transmission model with  a  path loss factor of two for the logarithm of sender and receiver energy ratio. We are aware, that this equation describes only the far-field behavior, which starts at some constant numbers $c_f>1$ of wavelengths, i.e. $|\!|q-v_i|\!|_2 \geq c_f \lambda$ (Antennas have unit size and are neglected in this equation). Hence, every time  $|\!|r-v_i|\!| < c_f \lambda$,
we will replace the denominator $|\!|q-v_i|\!|_2$ by  $c_f \lambda$ in this expression. We assume that $c_f \lambda \leq 1$ and therefore $\lambda < 1$.

\subsection{Unit Disk Graph Model}

For nodes $v_1, \ldots, v_n$, the geometric {\bf Unit Disk Graph} is defined by the set of edges $(v_i,v_j)$ where  nodes have distance $|\!|v_i,v_j|\!|_2 \leq 1$. In each round a message or signal can be sent from a node to an adjacent node. So, collaborative sending is simply ignored. Yet, we also ignore the negative effect of interference. In this model messages can be sent along edge in parallel, independently from what happens somewhere else.
     
The following Lemma shows the strong relationship between the single sender MIMO model and the UDG model.
\begin{lemma}
If only one sender $u$ sends a signal in the MIMO model with amplitude $a\in \REAL^+$, then a node $v$ in distance $d$ receives it if and only if $d\leq \frac{a}{\sqrt{\beta N_0}}$.
\end{lemma}
\begin{proof}
In the MIMO model the sender produces the signal $s_j(t) = a\  e^{i 2\pi t/ \lambda +i\phi}$ for $t \in [0,T]$ for some sending time $T$ and phase shift $\phi$ and otherwise $s_j(t) = 0$.
For distance $d$ the received signal is 
$$\hbox{rx}(t) =  \frac{s_j(t- d)}{d}  = \frac{a\  e^{i 2\pi (t-d)/ \lambda+i\phi}}{d}
\ ,$$ if $t \in [d,T+d]$ and otherwise $\hbox{rx}(t)= 0$.
Therefore for $\delta = T$:
\begin{eqnarray*}
z & = &  \frac1{T}\int\displaylimits_{t=d}^{d+T} \hbox{rx}(t) \  e^{-i 2 \pi t/ \lambda} \mathrm{d}t \\
& = &  \frac1{T}\int\displaylimits_{t=d}^{d+T}\frac{a\  e^{i 2\pi (t-d)/ \lambda+i\phi}}{d}\  e^{-i 2 \pi t/ \lambda} \\
& = &  \frac1{T}\int\displaylimits_{t=d}^{d+T}\frac{a\  e^{-i 2\pi d/ \lambda +i\phi}}{d}\  \mathrm{d}t \\
& = &  \frac{a}{d}\ \   e^{-i 2\pi d/ \lambda+ i \phi} \\
\end{eqnarray*}
Now $ |z| = a/d$ and therefore if $|z|^2 = \frac{a^2}{d^2} \geq \beta N_0$ then $u$ can receive the  signal.
\end{proof}
This Lemma implies that if $a^2= \beta N_0$, then the MIMO model is equivalent to the  Unit-Disk Graph (UDG)  model with  sending radius~1, if only one sender is active. In order to fairly  compare these two models, we fix $a=1$ and set  $\beta N_0 = 1$. 

%
\subsection{The Signal-to-Noise-Ratio Model}

The {\bf Signal-to-Noise-Ratio (SNR)} model adds the received signal energy of all senders, i.e. a signal is received at $q$ in the SNR model, if for sender energy $S_j := a_j^2$, where $a_j$ denotes the amplitude of sender $v_j$ the sum of the received signal energy is large enough: 
$$\hbox{\sl RS} := \sum_{j=1}^n \frac{S_j}{(|\!|q-v_j|\!|_2)^2} \ , \quad \hbox{where} \quad \quad \frac{\hbox{\sl RS}}{N_0} \geq \beta \ .$$
If we assume that the senders' starting time is not coordinated but independently chosen at random, then the following Lemma shows that the MIMO model in the expectation is equivalent to the SNR model.

\begin{lemma}
At the receiver $q$ the expected signal energy $S$ of senders $v_1, \ldots, v_n$ with random phase shift $\phi_i$ and amplitude $a_i$ in the MIMO model is $$\Exp{S} = \hbox{\sl RS} =  \sum_{j=1}^n  \frac{a_j^2}{(|\!|q-v_j|\!|_2)^2}  \ .$$
\end{lemma}
\begin{proof}
We assume that at the receiver $q$ senders have started sending such that the sending time intervals $[t_j, t'_j]$ of sender $v_j$ covers the time interval $[T,T+\delta]$ at $q$, i.e. $t_j+ |\!|q-v_j|\!|_2 \leq t$ and $t'_j |\!|q-v_j|\!|_2 \leq t+ \delta$. 
Each sender sends the signal  $s_j(t) = a_j\  e^{i 2\pi t/ \lambda +i\phi_j}$ in interval $[t_j, t'_j]$ with random phase shift $\phi_j$. 
The received signal at $q$ during $t\in [T,T+\delta]$ is by definition:
\begin{eqnarray*}
\hbox{rx}(t) &=&   \sum_{j=1}^n \frac{s_j(t- |\!|q - v_j|\!|_2)}{|\!|q-v_j|\!|_2}\\
& = &  \sum_{j=1}^n  \frac{a_j}{|\!|q-v_j|\!|_2} \  e^{i 2\pi (t-|\!|q-v_j|\!|_2)/ \lambda +i\phi_j}\ .
\end{eqnarray*}
So, the received signal $z$ is:
\begin{eqnarray*}
z &=&  \frac1{\delta}\int\displaylimits_{t=t_0}^{t_0+\delta} \hbox{rx}(t) \  e^{-i 2 \pi t/ \lambda} \mathrm{d}t \\
&=&  \frac1{\delta}\int\displaylimits_{t=t_0}^{t_0+\delta}
 \sum_{j=1}^n  \frac{a_j}{|\!|q-v_j|\!|_2} \  e^{i 2\pi (t-|\!|q-v_j|\!|_2)/ \lambda +i\phi_j}\  e^{-i 2 \pi t/ \lambda} \mathrm{d}t \\
 &=&  \frac1{\delta}\int\displaylimits_{t=t_0}^{t_0+\delta}
 \sum_{j=1}^n  \frac{a_j}{|\!|q-v_j|\!|_2} \  e^{-i 2\pi |\!|q-v_j|\!|_2/ \lambda +i\phi_j} \mathrm{d}t \\
&=&  \
 \sum_{j=1}^n  \frac{a_j}{|\!|q-v_j|\!|_2} \  e^{-i 2\pi |\!|q-v_j|\!|_2/ \lambda +i\phi_j}  \\
 &=&  \ 
 \sum_{j=1}^n  b_j \  e^{i\sigma_j} \ , \\
\end{eqnarray*}
where we substitute $b_j= \frac{a_j}{|\!|q-v_j|\!|_2} $ and $\sigma_j=-2\pi |\!|q-v_j|\!|_2/ \lambda +\phi_j \bmod 2 \pi$. 
Note that $\sigma_1, \ldots, \sigma_n$ are again independent random variables and uniform distributed over $[0,2\pi]$. Now, we observe.
\begin{eqnarray*}
\Exp{\left| \sum_{j=1}^{n} b_j e^{i \sigma_j} \right|^2 } &=&
\Exp{\left(\sum_{j=1}^{n} b_j e^{i \sigma_j} \right) \cdot  \left(\sum_{j=1}^{n} b_k e^{- i \sigma_j} \right) } \\
&=&
\Exp{\sum_{j,k \in \{1, \ldots,  n\}} b_j b_k\   e^{i (\sigma_j - \sigma_k)}  } \\
&=&
\sum_{j,k \in \{1, \ldots,  n\}}  b_j b_k\   \Exp{e^{i (\sigma_j-\sigma_k)}  } \\
&=&
\sum_{j \in \{1, \ldots,  n\}}  b_j^2\   \Exp{e^{i (\sigma_j - \sigma_j)}  }
+
 \sum_{j,k \in \{1, \ldots,  n\}, j\neq k}  b_j b_k\  \Exp{e^{i \sigma_j}} \Exp{e^{-i \sigma_k}  } \\
 &=&
 \sum_{j,k \in \{1, \ldots,  n\}} b_j^2\ ,
\end{eqnarray*}
where we use that $\sigma_j$ and $\sigma_k$ are independent and that $ \Exp{e^{i \sigma_j}  } =0$ because $\sigma_j$ is uniform over $[0, 2\pi]$.
Therefore $\Exp{|z|^2} =  \sum_{j=1}^n  \frac{a_j^2}{(|\!|q-v_j|\!|_2)^2}$.
\end{proof}
This proof can  also be found in \cite{2012RandomMIMOJanson}.  Unlike in the coordinated MIMO model, in the SNR model signals are sent with random phasing which induces a more regular radiation pattern.

Under the assumption that $\Exp{|z|^2}/N_0 \geq \beta$ induces a successful reception, $a_j= 1$ and  $\beta N_0 = 1$ we derive the  Signal-to-Noise Ratio (SNR) model, where the energy of the uncorrelated received signals add up. Again, this model reduces to the UDG model if only one node is sending.

\section{UDG Coverage} \label{sec:UDG}

Let $\rho =\frac{n}{\pi R^2}$ denote the density of nodes. The probability
that $n-1$ nodes are not in a given area of size $\pi/8$ is 
$$
\left(1-\frac{\pi/8}{\pi R^2}\right)^{n-1} = 
\left(1-\frac{\rho \pi}{8 n}\right)^{n-1} \leq \exp\left(-\left( 1-\frac1n\right) \frac{\rho \pi}{8} \right)\ ,$$ which is less than $1/n^{c}$ for $\rho \geq \frac{8}{\pi} c \ln (n+1)$ for any $c>1$ and $n\geq 2$. 
Now, consider six equally sized  sectors of a unit disk around a node; there is at least a node with probability $1 - n^{-c+1}$. 

\begin{figure}[ht]
\centerline{\includegraphics[width=0.25\linewidth]{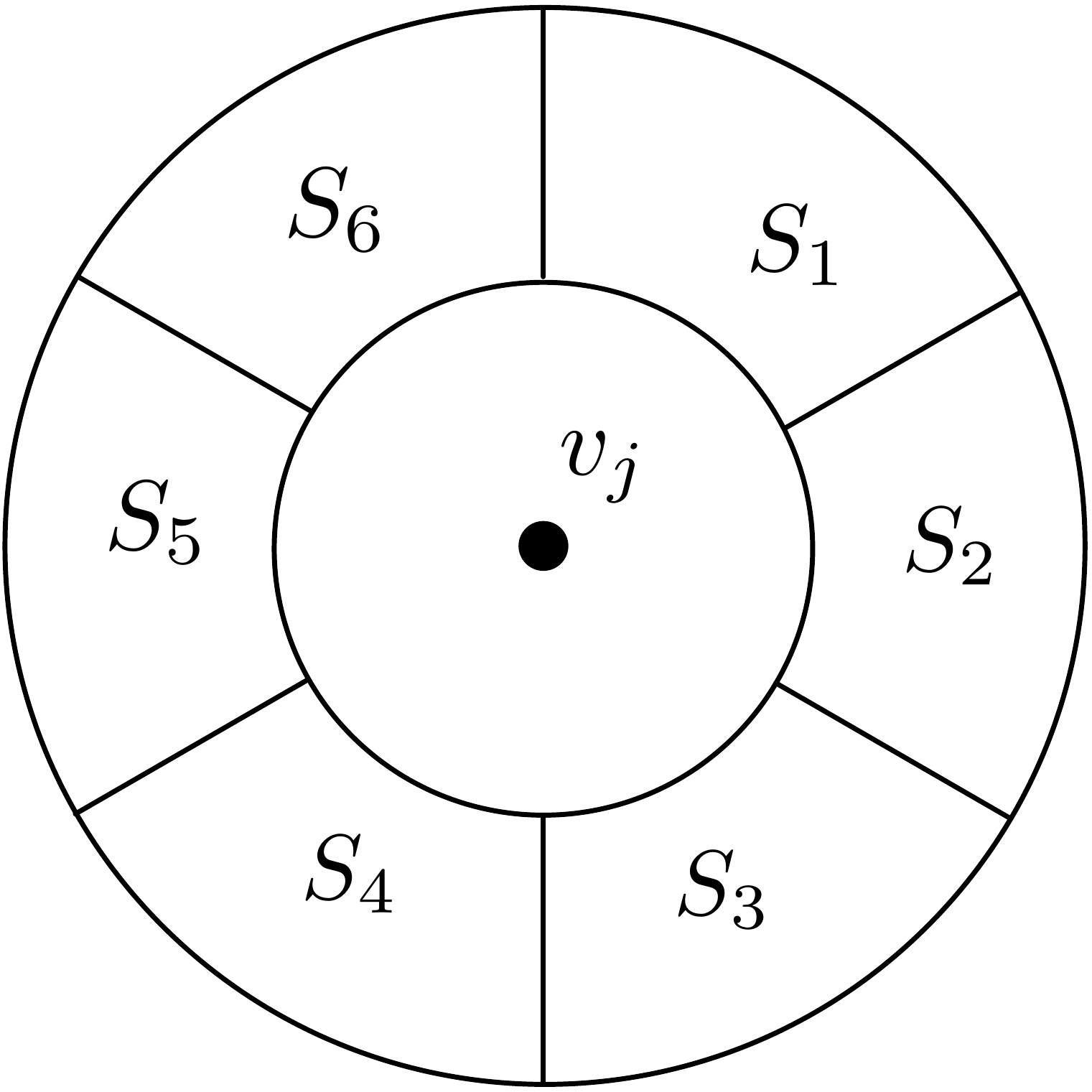}}
\caption{For a disk with radius $1$ and area $\pi$ and node density $\rho = \Omega(\log n)$, each of the sectors $S_1, \ldots, S_6$ around a disk with area $\pi/8$ and sender $v_j$ in the center are not empty with high probability\label{sectors}}
\end{figure}

From this, it follows that UDG is connected (see~\cite{xue2004number} for a better bound) and that the diameter of the UDG is at most $8R = \mathcal{O}\left(\sqrt{n/\rho}\right)$.

\begin{lemma}
For $\rho>1$ in the UDG model, broadcasting needs $\Omega(\sqrt{n/\rho})$ rounds to inform all nodes with high probability.
\end{lemma}
\begin{proof}
The probability that none of the $n-1$ non centered nodes are at a distance larger than $R-1$ from the center is for $R>1$:
$$ \left(1- \frac{2R-1}{R^2} \right)^{n-1} \leq e^{-\frac{n-1}{R}} = e^{-\Theta(\sqrt{n \rho})}. $$ Hence, with high probability some nodes are in this outer rim, which can be reached only after at least $R-2 = \Omega(\sqrt{n}{\rho})$ rounds. \qed
\end{proof}
For large enough density $\rho=\Omega(\log n)$ this bound is tight.
The probability
that $n-1$ nodes are not in a given area of size $\pi/8$ is 
 less than $1/n^{c}$ for $\rho \geq \frac{8}{\pi} c \ln (n+1)$ for any $c>1$ and $n\geq 2$. 
%
%
\begin{theorem}
For $\rho= \Omega(\log n)$ in the UDG model, broadcasting needs $\Theta(\sqrt{n/\rho})$ rounds to inform all nodes with high probability.
\end{theorem}
\begin{proof}
Consider two nodes $v_j$ and $v_k$ with distance $d\leq R$. We have seen that each subregion around a node depicted in Fig~\ref{sectorrouting}  contains at least a node with high probability. Now, we route starting from $v_j$ along the line $L$ connecting $v_j$ and $v_k$ by choosing a node from a sector which is closer to $r_k$ in a sector which in a corridor of width $2$ around $L$. We pick a node from this sector and observe that the messages advances by a distance of at least $\frac14$ in the direction towards $v_k$.
\begin{figure}[ht]
\centerline{\includegraphics[width=0.85\textwidth]{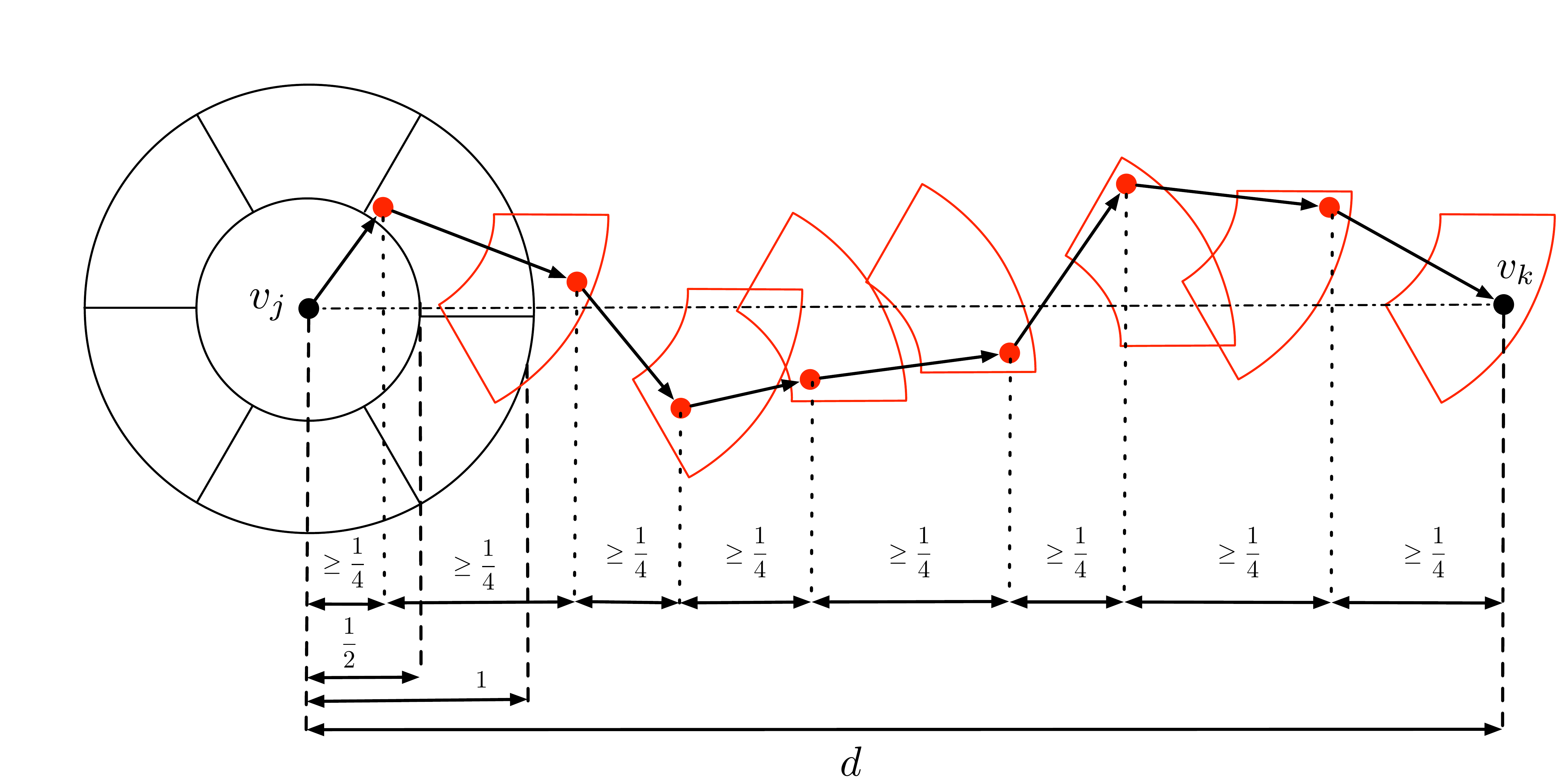}}
\caption{Routing from $v_j$ to $v_k$ using the unit-disk graph and non-empty sectors.\label{sectorrouting}}
\end{figure}

Hence, it takes at most $4R$ hops, where $R^2 = \frac{n}{\pi \rho}$.
\qed\end{proof}

\section{A Lower Bound for SNR Collaborative Broadcasting} \label{sec:SNR}

The  expected number of nodes $n(r)$ in a disk of radius $r$ around the origin is sharply concentrated around the expectation $\rho \pi r^2$, if it is at least logarithmic in~$n$, which follows from an application of Chernoff bounds.
\begin{lemma} \label{densifity}
For $n$ randomly distributed nodes in a disk of radius $R$ and a given smaller disk of radius $r$ within this disk, let $n(r)$ denote the number of nodes there within. Then we observe: 
\begin{eqnarray}
\Exp{n(r)} &=& 
\pi \rho r^2\ ,  \\
\Prob{n(r) \geq (1+c) \Exp{n(r)}}& \leq & e^{-\frac13 \min\{c, c^2\} \pi \rho r^2}
\ .  
\\ \Prob{n(r) \leq {\textstyle \frac12}  \Exp{n(r)}}& \leq & e^{-\frac18   \pi \rho r^2} \ . \label{dlower}
\end{eqnarray}
\end{lemma}
\begin{proof}
We can reformulate $n(r)$ as the sum of the independent Bernoulli variables $X_i$, which denote $X_i=1$ when node $i$ falls into the smaller disk of radius $r$, and otherwise $X_i=0$. We have  $\Prob{X_i=1}= \frac{\pi r^2}{\pi R^2} = \frac{r^2}{R^2}$ and thus the expectation of $n(r)$ is the following.
$$\Exp{n(r)} = \Exp{\sum_{i=1}^n X_i} = n \frac{r^2}{R^2} = \frac{\pi \rho n r^2 }{n}=  \pi\rho r^2\ ,$$
using $\rho = \frac{n}{\pi R^2}$. The other inequalities follow by applying Chernoff bounds to $n(r)$.
\end{proof}

\begin{theorem} In the SNR-model for $\pi \rho \geq 1$ and $\rho =o(n)$  at least $\Omega\left(\frac{\log n}{\max\{1,\log \rho\}}\right)$ rounds are necessary to broadcast the signal to all $n$ nodes  with high probability.
\end{theorem}
\begin{proof}
We start with the center node in the middle of the disk and denote by $r_j$ the maximum distance of an informed node from the center of the disk. Let $n_j$ denote the number of informed nodes in round $j$. By definition $r_0 = 0$ and $n_0=1$. Then, in round one we have $r_1=1$ by applying the SNR model for one sender.

We consider two cases.
\begin{enumerate}
\item Case: $\pi \rho \geq k \log n$. 

Then, the expected number of nodes $n_1$ is $\pi \rho$ by Lemma~\ref{densifity} and for $\pi \rho \geq 1$ it is bounded as  $n_1\leq 2\pi \rho$ with high probability by  choosing $c=3k$. 
Consider a receiver in distance $d$ and assume for the lower bound argument that all nodes $n(r_j)$ in radius $r_j$ send the signal. Since $\pi \rho r_j^2 \geq k \log n$ we have $n(r_j) \geq \frac12 \pi \rho^2$ with high probability.
So, for $d  \geq 4 \sqrt{\rho} r$ and $\rho \geq 1$ we 
have $$d-r  
\geq 4 \sqrt{\rho} r-r  
\geq (4 \sqrt{\rho}-1)r
\geq (4 \sqrt{\rho}-\sqrt{\rho})r
 \geq 3 \sqrt{\rho} r 
 > \sqrt{2 \pi \rho} r \ .$$ 
 Then, the received energy is at most $\frac{n(r)}{(d-r)^2}$ where
$$ \frac{n(r)}{(d-r)^2} < \frac{n(r)}{2 \pi \rho r^2} \leq \frac{2 \pi \rho r^2}{2 \pi \rho r^2}  \leq 1\ \ ,$$ 
with high probability. 
So, no node farther away than $r_{j+1} = 4 \sqrt{\rho} r_j$ is informed in the SNR model in round $j$.

By induction only nodes in distance of at most $r_t = \left(4 \sqrt{\rho}\right)^t$ can be informed  after $t$ rounds with probability larger than $\frac1{n^{\mathcal{O}(1)}}$, which only can inform all nodes outside the disk of radius $R-1= \frac{n}{\pi \rho}-1$ if  
$ t\geq 
\Omega\left(\frac{\log n}{\log \rho}\right)$ for $\rho=o(n)$.

\item $\pi \rho\geq 1$ and $\pi \rho \leq k \log n$. 

For the proof we have to overcome the difficulty that the number of nodes in the unit disk may be too small to ensure high probability. We resolve this problem by overestimating the first radius $r_1=\sqrt{\frac{k}{\pi\rho} \ln n}$. Then, the expected number of nodes in this disk is $\Exp{n(r_1)} = \pi \rho r_1^2 = 3 k \ln n$ and $n(r_1)\geq 2 \Exp{n(r_1)}$ with small probability, i.e. $1/n^k$. 

Like in the first case we assume that in round $r_j$ all nodes in this radius send. So, for $d  \geq 4 \sqrt{\rho} r$ and $\rho \geq 1$ the received energy is less than $1$ within a distance of at most $r_{i+1} =  4 \sqrt{\rho} r_j$. 

Now the recursion is $$r_t = (4\sqrt{\rho})^{t-1} r_1 = (4\sqrt{\rho})^{t-1} \sqrt{\frac{k}{\pi\rho} \ln n} = 4 (4\sqrt{\rho})^{t-2} \sqrt{\frac{k}{\pi} \ln n} \ .$$

After $t$ rounds nodes in distance of at most $r_t$ can be informed, which can inform all nodes in the disk of radius $R= \frac{n}{\pi \rho}$ if  
$$ 4 (4\sqrt{\rho})^{t-2} \sqrt{\frac{k}{\pi} \ln n}  \geq R = \frac{n}{\pi \rho}\ ,$$ yielding 
$$  t \geq 2 + \frac{\log n - \frac12 \log \log n  - \frac12\log \pi - \log \rho + \log k +2}{2+ \frac12 \log \rho} = \Omega\left(\frac{\log n}{\max\{1,\log \rho\}}\right)$$ since $\rho=o(n)$ and $\rho\geq 1$.
\end{enumerate}
\end{proof}

\section{A Lower Bound for MIMO Collaborative Broadcasting} \label{sec:MISO_Lower}
If the unit length amplitudes of all senders in a disk of range $r$ are superpositioned,  in the best case this results in a received absolute amplitude proportional to the number of senders divided by the distance.
\begin{lemma}\label{misolowerbound}
Assuming that randomly placed senders are in a disk of radius $r$, then the maximum distance of a node which can be activated is at most 
$4\pi \rho r^2$ with high probability for  $\rho r^2 = \Omega(\log n)$.
\end{lemma}
\begin{proof} The expected number of senders in a disk of radius $r$ is $\pi \rho r^2$. Using Chernoff bounds and $\rho r^2 = \Omega(\log n)$ one can show that this number does not exceed $2 \pi \rho r^2$ with high probability.

Now, in the best case, all waves at a receiver $r$ perfectly add up resulting in a received signal of at most 
$\left|\hbox{rx}\right|\  \leq \  \sum_{i=1}^{2 \pi \rho r^2}
  \frac{1}{|\!|r-s_i|\!|_2}\ .$ We overestimate this signal by replacing the
  denominator with $d-r$, where $d$ is the distance of the receiver from the senders' disk's center. 
Hence, we receive a signal if $\left|\hbox{rx}\right|^2 = (2 \pi \rho r^2)^2 \geq  (d-r)^2$. So, we get $d \leq r+ 2\pi \rho r^2 \leq 4 \pi \rho r^2$.
\end{proof}
This Lemma implies the following lower bound.
\begin{corollary}
Any broadcast algorithm using MIMO needs at least $\Omega(\log \log n - \log \log \rho)$ rounds to inform all $n$ nodes with high probability.
\end{corollary}
\begin{proof}
	We use Lemma~\ref{misolowerbound} by overestimating the effect of triggered nodes which are bound to disks with radii $r_j$. We assume that we start  with $r_0 = \log n$ for $\rho\geq 1$. Now, let $r_{j+1} = 4\pi \rho r_j^2$ denote the largest distance of a node in the next round.
	
So $r_j \leq  (4 \pi \rho \log n)^{2^j}$, which reaches $R-1=\sqrt{n/(\pi \rho)}-1$ at the earliest for some $j= \Omega(\log \log n - \log \log \rho)$. 
   \end{proof}

This claim also follows from the considerations in
\cite{JS13_Beamforming_Line} and \cite{JS14_Beamforming_LogLog_SSS} and
more extensive in \cite{diss-janson-2015} where
 a lower bound of $\Omega(\log \log n)$ rounds for the unicast problem has been shown. Here, we adapt this argument to include the density $\rho$.
 
\section{Expanding Disk Broadcasting} \label{sec:SNR_EDB}
For the SNR model a simple flooding algorithm works as well as the algorithm we propose. A straight-forward observation is a monotony property, i.e. every increase in sending amplitude and every additional sending node increases the coverage area. For the upper bound we use Algorithm~\ref{alg:EDB} which is slower, yet still asymptotically tight to the lower bound and easier to analyze. We choose $r_{j+1} = \frac14 \sqrt{\rho} r_j$, starting with $r_1=1$ and prove the following Lemma.

\begin{lemma}If $\rho= \Omega(\log n)$, then
in round $j\geq 1$ all nodes in distance $r_{j+1}$ from the origin have been informed with high probability.
\end{lemma}
\begin{proof}
Lemma~\ref{densifity} states that the expected number of nodes $n(r_j)$ in the disk of radius $r_j$ is $\rho \pi r_j^2$. Lemma~\ref{densifity} (\ref{dlower}) shows that $\Prob{n(r_i)\leq \frac12 \pi \rho r_j^2} \leq e^{-\frac18 \rho \pi r_j^2} \leq e^{-\frac18\rho}$, which is a small probability $1/n^c$ for $\rho= \Omega(\log n)$.

The maximum distance from any node in the disk of radius $r_{j+1}$ to a node in this disk is at most $r_j + r_{j+1} \leq 2 r_{j+1}$. Hence, the received signal has an expected SNR of at least 
$$
\frac{n(r_j)}{(2r_{j+1})^2} \geq 
\frac{\frac12 \rho \pi r_j^2}{(2r_{j+1})^2} = \frac{\frac12 \rho \pi r_j^2}{(2 \frac14 \sqrt{\rho} r_j)^2}  = 2\pi > \beta = 1 \ .$$
\end{proof}
\SetEndCharOfAlgoLine{}

\begin{algorithm}
\caption{Expanding Disk Broadcast}
\label{alg:EDB}
\SetKwProg{generate}{Algorithm \emph{Expanding Disk Broadcast}}{}{end}

\generate{}{
      Sender $v_0$ starts sending \;
      $j \leftarrow 1$ \;
       \While{$r_j<R$}{
     \For{\textbf{all}  $v\in\{v_1,\ldots, v_n\}$ which are informed and where $|\!|v - v_0 |\!|_2 \leq r_j$}{
    Node $v$ starts sending \;
      }
      $j\leftarrow j+1$\;
     }
}
\end{algorithm}
Therefore $r_j= (\rho/16)^{(j-1)/2}$ and  for 
$j
\geq 1+ 2\frac{\log  n - \log( \pi \rho)}{(\log \rho) - 4}
= \Theta(\log n/\log \rho)$ we have $r_j \geq R$ and all nodes are informed. 
\begin{corollary} \label{co:SNR_log}
In the SNR-model collaborative broadcasting needs $ \mathcal{O}(\log n/\log \rho)$ rounds for $\rho > 16$, if broadcasting starts with at least $\Omega(\log n)$ nodes, or $\rho=\Omega(\log n)$.
\end{corollary}

We conjecture that the result of Corollary~\ref{co:SNR_log} not only holds for  our (line-of-sight, path loss exponent 2) SNR model but also holds for the model proposed in \cite{NGS09_Linear_Capacity_Beamforming,ozgur2007hierarchical} where the path loss exponent is $\alpha \le 2$. Then, the channel from sender $v_j$ to receiver $v_k$ has an contribution of $s_j(t) h_{j,k}(t)  $ for emitted signal $s_j(t)$ and $h_{j,k}(t) = |\!|v_k-v_j|\!|_2^{-\alpha/2} \cdot e^{i\cdot \theta_{j,k}(t)}$ with random phase shift $\theta_{j,k}(t)$ at time $t$. We  discuss further conjectures about the influence of the path loss factor in the Outlook.

The transmission range grows in each round exponentially with factor $r_j / r_{j-1} = \sqrt{\rho/16}$. In order to minimize this term, we define the radii downwards. Let $p = \lceil2\log R/(\log \rho/16)\rceil$ denote the last round and define $r'_p := R = \sqrt{\frac{n}{\rho \pi}}$. Then, 
define $r'_{j-1} := r'_j/\sqrt{\rho/16}$. Note that $r'_j \leq r_1 =1$ and therefore $r'_2\leq r_2$ can be informed in one step.

So, the overall time $T$ is mostly determined by the speed of light:
\begin{eqnarray*}
T &\le& \sum_{j=1}^{p} r'_j  
\ \le\  R \left( 1 + \sum_{j=1}^{p-1} \frac{1}{(\log \rho/16)^{j/2}}\right)
= R \left(1 + \mathcal{O}\left(\frac{1}{\sqrt{\log \rho}}\right)\right)  \ .
\end{eqnarray*}
Since we assume that $\rho= \Omega(\log n)$ we have in this setting the following corollary.
\begin{corollary}\label{c:snrlight}
The speed of the SNR Broadcast approaches the speed of light for growing $n$.
\end{corollary}
However, the time effort for decoding and encoding at a node is for practical applications usually much larger than the time caused by the speed of light, the main factor is the number of times the signals have to be relayed, i.e. the number of rounds.

\section{MIMO} \label{sec:MISO}

In MIMO we only analyze the expanding broadcasting algorithm, since the coverage area is far from convex nor does every additional sending node help, see Fig.~\ref{figspikes}. We use a start radius $r_1= c_{2}/\lambda$ and the expansion $ r_{j+1} = c_1  \rho r_j^{3/2}\lambda^{1/2}$ for a constant $c_1>0$ to be defined later. 
\begin{figure}[ht]
\centerline{\includegraphics[width=1\linewidth]{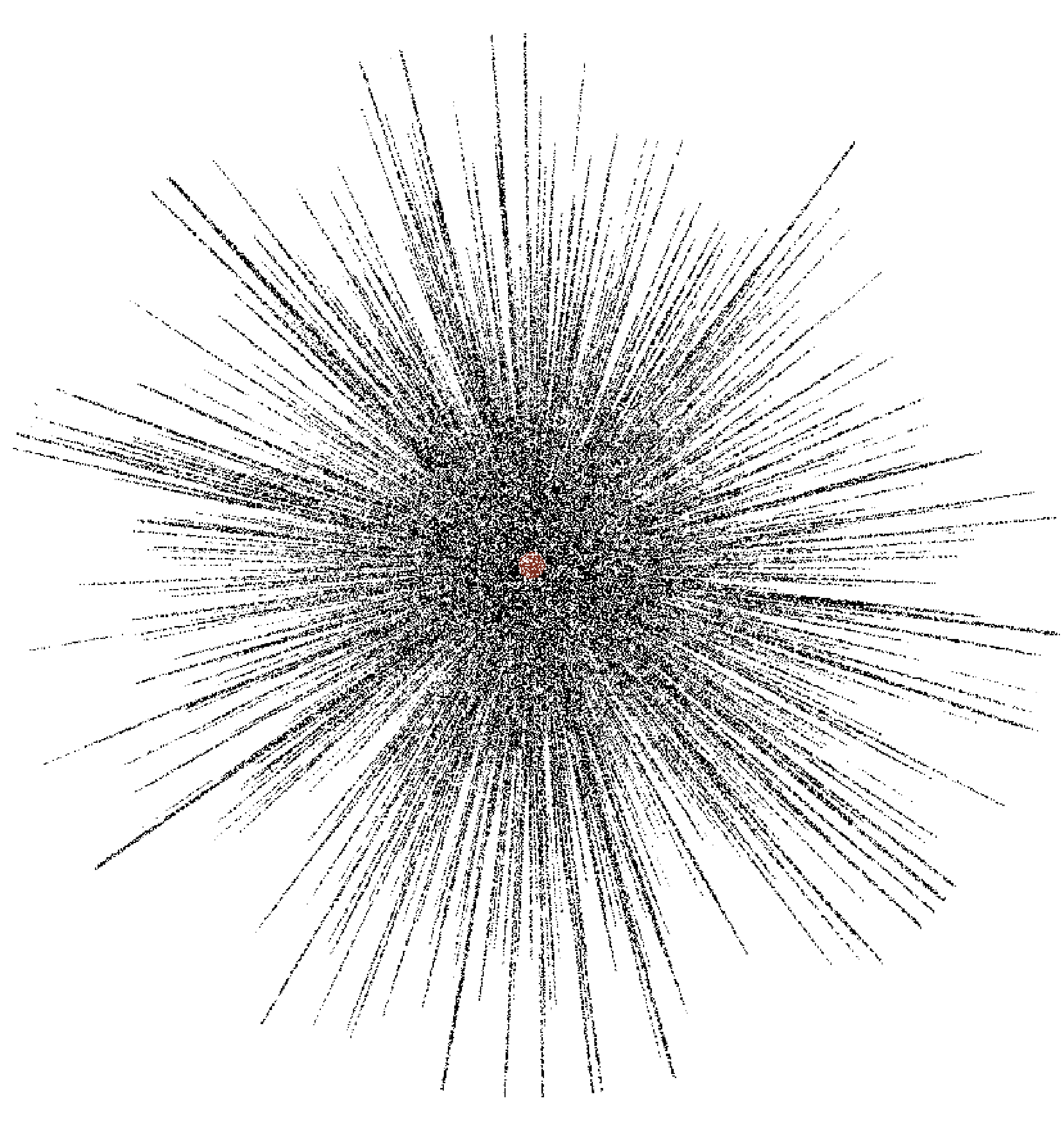}}
\caption{Randomly placed nodes reached by a cooperative broadcast in the MISO model from the set of senders randomly placed in the red disk.\label{figspikes}}
\end{figure}

Informing all nodes in radius $r_1= c_{2}/\lambda$ is done by a different algorithm, since the MISO broadcast does not work here. For constant $\lambda$ this can be done with single sender broadcasts (resulting in the UDG broadcast). For $\lambda=o(1)$ the number of rounds in the first phase may have a significant impact. 

In the subsequent MIMO rounds, the senders $v_k$ are synchronized with a phase shift $\varphi_\ell = - 2 \pi |\!|v_\ell - v_0|\!|_2 / \lambda$ such that the resulting signal of $v_\ell$ is $e^{i (2 \pi t / \lambda + \varphi_\ell)}$. These phases try to imitate the pattern of single sender in the center, the energy of which grows double exponentially in each round. 

For the analysis we consider only the signal strength at one receiver and analyze whether MISO works for this sender. We prove that the SNR ratio of the collaborative broadcast signal at every receiver is above the threshold with high probability. So,  MISO with high probability results in MIMO with high probability for all receivers in the next disk rim.

\begin{algorithm}
\caption{MISO Broadcast}
\label{alg:MISOB}
\SetKwProg{generate}{Algorithm \emph{MIMO Broadcast}}{}{end}

\generate{}{
       Inform all nodes in the disk of radius $15 r_1$\;
       $j \leftarrow 1$ \;
       \While{$r_j<R$}{
     \For{\textbf{all}  $v\in\{v_1,\ldots, v_n\}$ which are informed and where $|\!|v - v_0 |\!|_2 \leq r_j$}{
    Node $v$ starts sending with phase shift $\varphi = - 2 \pi |\!|v - v_0|\!|_2 / \lambda$ \;
      }
      $j\leftarrow j+1$\;
     }
}
\end{algorithm}

Recall that the density is defined as $\rho= \frac{n}{\pi R^2}$ and let $\rho \geq c_{3} \log n$. 

\begin{theorem}\label{safety}
For constant wavelength $\lambda$, density $\rho =  \Omega(\log n)$  every receiver in distance $d$ 
can be triggered with high probability, if $ 15r_j  \leq d  \leq  c_1 \rho r_j^{3/2}\lambda^{1/2} $, for a constant $c_1$.
\end{theorem}
\begin{proof}
We consider an arbitrary node $q$ in distance $d$ from the first sender $v_0$ in the center.
We  prove that this node is triggered with high probability and thus all receivers in this distance will be triggered likewise with this probability. 

First we analyze the expected received signal of a receiver in distance $d$, which is given by an integral. The complex value of this integral will be asymptotically estimated using a geometric argument over the intersection of ellipses with equal phase shift impact and the sender disk.

In this proof $m+1$ denotes the number of triggered senders in a disk of radius $r$. Sender $v_0$ resides at $(0,0)$. The other $m$ senders are uniformly distributed in the disk of radius $r$. We investigate the received signal $\hbox{\sl rx}$ at a receiver $q$ outside of this disk with distance $d\geq r$. Wlog. we assume $q$ lies on the $x$-axis.

Define for $p=(p_x,p_y)$:
\begin{equation}
\Delta_d(p) := \sqrt{p_x^2+p_y^2}+ \sqrt{(d-p_x)^2+p_y^2}-d \ . \end{equation}
Using this notion the  received signal strength is given by

\begin{lemma}
For  $0\leq w \leq \tau+\lambda/2$ and senders $v_1, \ldots, v_n \in D(v_0,r)$ the received signal is given as
$\hbox{\sl rx}(t)= \hbox{\sl rx} \cdot 
e^{i 2 \pi (t- d)/\lambda}$, where $\hbox{\sl rx} =
\sum_{j=1}^n \frac{ e^{-i 2 \pi
 \Delta_d(v_j)/{\lambda}} }{|\!|q-v_j|\!|_2}  \ .$
\end{lemma}
We will estimate the absolute value $|\hbox{\sl rx}|$ as follows.

First, we see that there is an easy characterization by ellipses $E_{\tau}$ with focal points in $v_0$ and $q$, which characterize whether senders help or interfere, see Fig.~\ref{oak-senders}. The parameter $\tau$ describes the phase at which the sender's signal arrives at the receiver $q$. The main contributor to the received signal comes from the area within $E_{\tau}$ which has an area of $\Theta(r^{3/2} \lambda^{1/2})$, which corresponds to the innermost dark ellipse in Fig.~\ref{oak-senders}. The other areas more or less cancel themselves out.
 
 \begin{figure}[ht]
\centerline{\includegraphics[width=0.8\linewidth]{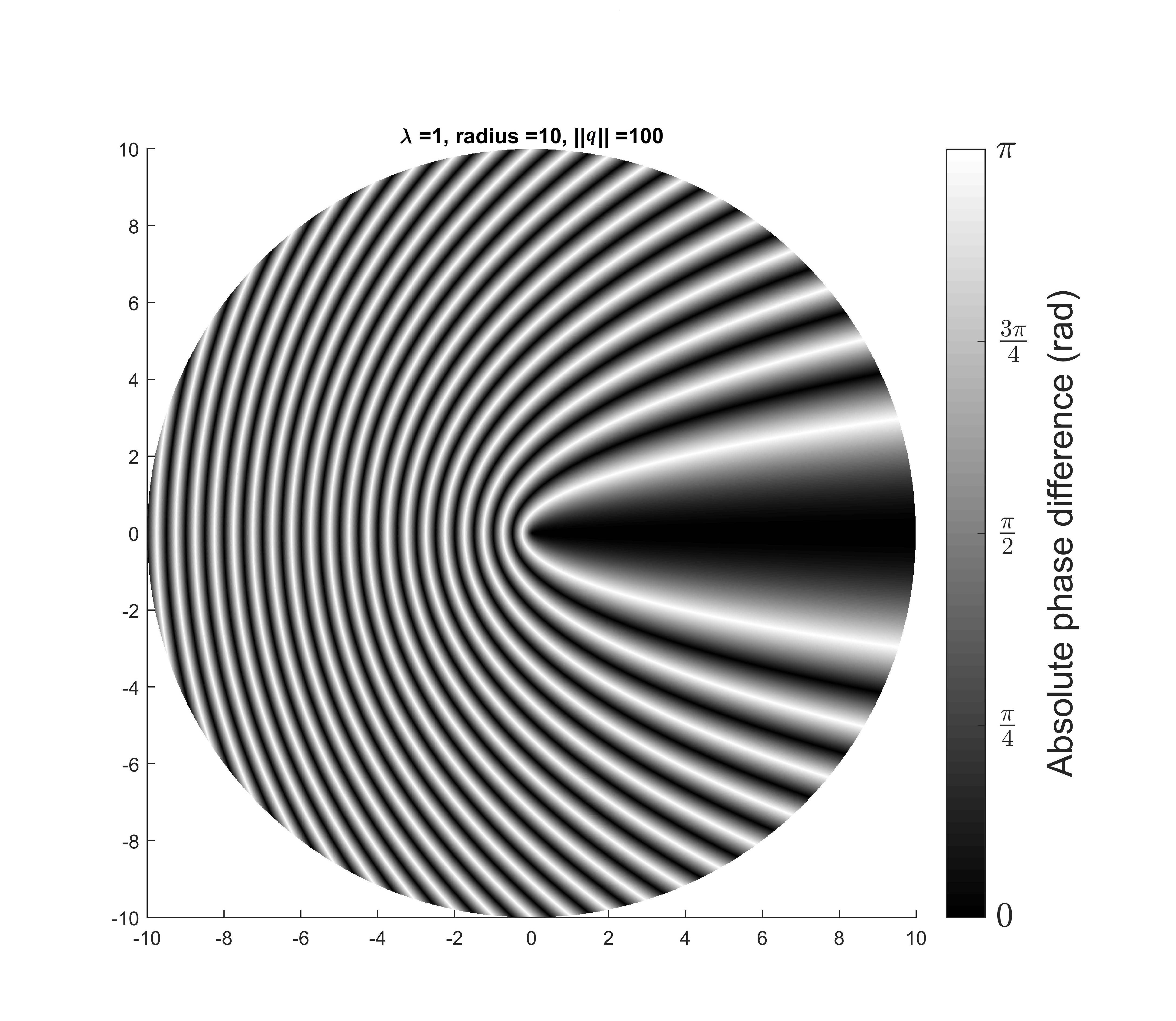}}
\caption{Senders in a disk of radius 10, colored according to the phase difference perceived by a receiver located at point (100, 0) for wavelength~1 wavelength $\lambda=1$ \cite{oakmasterthesis}.\label{oak-senders}}
\end{figure}
 
To prove this, we give a formula which describes exactly the expected signal at a given point $t_0$. This expectation will be estimated by carefully chosen bounds. We denote by $D_r$ the disk with center $(0,0)$ and radius $r$.

Consider the ellipse $E_{\tau}$ with focal points $v_0=(0,0)$ and $q= (d,0)$:
$$ E_{\tau} := \{ p \in \REAL^2 \mid |\!|p|\!|_2 + |\!|p-q|\!|_2 = d+ \tau \} \ .$$
We will use the following notation do describe all points inside this ellipse:
$$ E_{\leq \tau} := \{ p \in \REAL^2 \mid |\!|p|\!|_2 + |\!|p-q|\!|_2 \leq d+ \tau \} \ .$$
\begin{lemma}
$
\Exp{\hbox{\sl rx}} 
= \frac{1}{d} +  
\frac{m}{2\pi r^2}
\iint\displaylimits_{(x,y) \in  D_r} 
\frac{e^{- i \Delta_d(x,y) 2\pi/\lambda} \ \mathrm{d} x \ \mathrm{d} y}{\sqrt{(x-d)^2+y^2}} \ .
$\end{lemma}

In order to estimate this expectation we introduce the following complex valued functions for the
 disk $D_r = \{p \in \REAL^2 \mid |\!|p|\!|_2 \leq r\}$.
\begin{eqnarray}
s({d,\lambda,r}) &:=&\iint\displaylimits_{(x,y) \in  D_r} 
\frac{e^{- i \Delta_d(x,y) 2\pi/\lambda} \ \mathrm{d} x \ \mathrm{d} y}{\sqrt{(x-d)^2+y^2}} 
\end{eqnarray}
Since the maximum phase shift in the disk of radius $r$ is at point $(-r,0)$ with $\Delta_d(-r,0) = 2r$ we 
have the following relationship between these functions $s({d,\lambda,r}) = h_{d,\lambda,r}(2r)$. Furthermore,
there is the following linearity within $h$:
\begin{equation} \forall c>0: s({d,\lambda,r}) = \frac1c\ s({c d,c \lambda,c r})\end{equation}

Note that the following relationship holds.
\begin{eqnarray*} 
\Exp{\hbox{\sl rx}} &=&   \frac1d + \frac{m}{2\pi r^2} s(d,\lambda,r) \\
&=&  \frac1d + \frac{m}{2\pi r}s(d/r,\lambda/r,1) \ .
\end{eqnarray*}
Now the estimation of the signal is based on the following lemma.
\begin{lemma}
\label{achtvier}
For $w\geq \pi$,  $\lambda \leq 2$ and $d>15$:
\begin{eqnarray} 
\Im( s(d,\lambda,1)) &\geq& \frac{9}{2,240\sqrt2} \frac{\sqrt{\lambda}}{d+1} \ . 
\end{eqnarray}\end{lemma}
\begin{proof}
We need the following definitions, where $h_{\infty,\lambda}(w)$ is used to estimate the signal energy in the far-distance case. We first estimate its size and then apply it to $s(d,\lambda,r)$. Note that $U=D_1$.
\begin{eqnarray}
h_{d,\lambda}(w) &:=&  \iint\displaylimits_{(x,y) \in E_{\leq w} \cap U} 
e^{i \Delta_d(x,y) 2\pi/\lambda}\  \mathrm{d} x \ \mathrm{d} y\ , \\
h_{\infty,\lambda}(w) &:=& \iint\displaylimits_{(x,y) \in E_{\leq w} \cap U} 
e^{i \Delta_{\infty}(x,y) 2\pi/\lambda} \ \mathrm{d}x \ \mathrm{d}y \ .
\end{eqnarray}
where $
\Delta_d(x,y) := \sqrt{x^2+y^2}+ \sqrt{(d-x)^2+y^2}-d$ and $\Delta_{\infty}(x,y) = \sqrt{x^2+y^2} -x $.

We use a geometric argument and concentrate on the area of the intersection of the unit disk $U$ and the $E_{\leq w}$ described by
$f(w,d)= \iint\displaylimits_{(x,y)  \in E_{\leq w} \cap U} 1 \  \mathrm{d} x \ \mathrm{d} y 
$. Using the following function for the area of the segment of depth $z$ of an ellipse with radii $r_1$ and $r_2$ we derive a closed form for this function, see Fig.~\ref{segment}.
\begin{eqnarray*} 
g(x) & := & \arccos\left(1-x\right)-\left(1-x\right) \sqrt{1-(1-x)^2}\\
s(r_1,r_2,z) &=& r_1 r_2\ g\left(\frac{z}{r_1}\right) 
\end{eqnarray*}
\begin{figure}[ht]
\centerline{\includegraphics[width=0.35\linewidth]{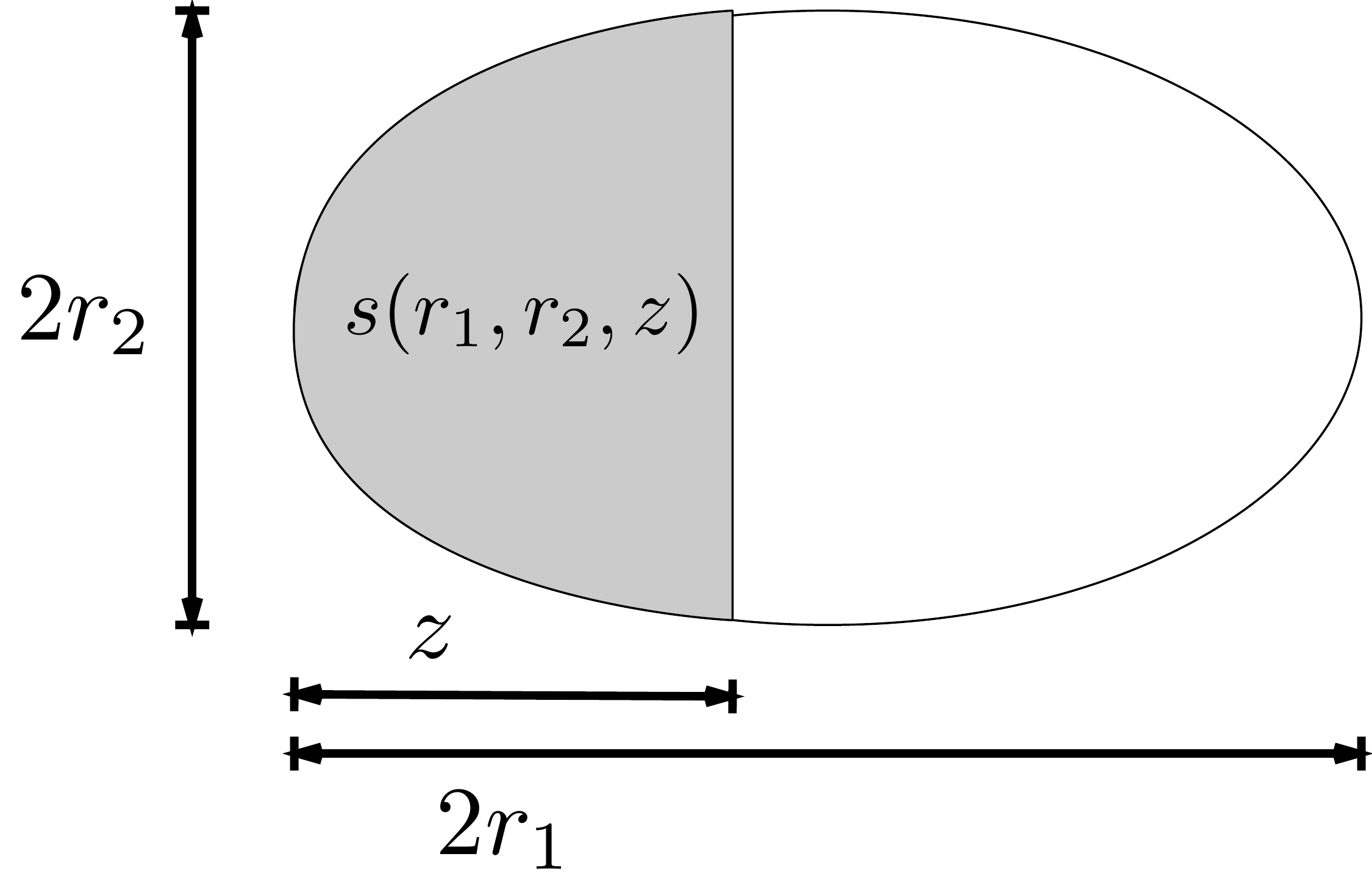}}
\caption{Area $s(r_1,r_2,z)$ of a segment of an ellipse \label{segment}}
\end{figure}

The main notations are shown in Fig~\ref{f:not}.
\begin{figure}[ht]
\centering
\includegraphics[width=.6\linewidth]{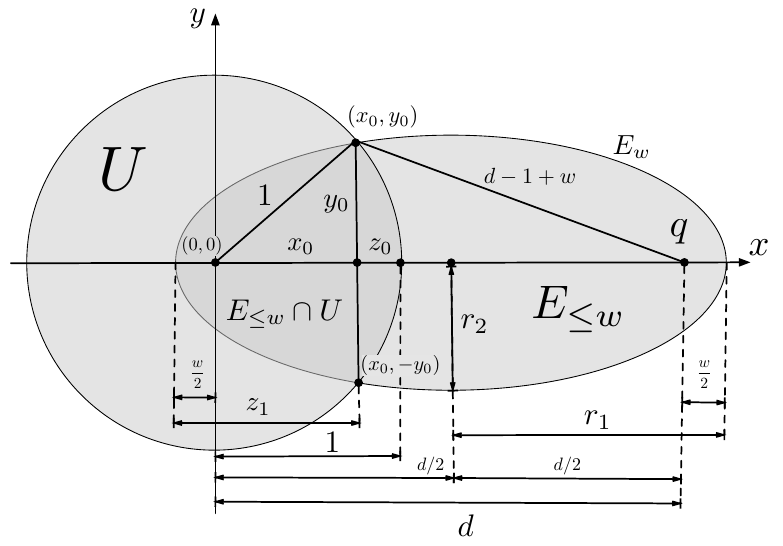}
\caption{Sender $v_0=(0,0)$, receiver $q=(d,0)$, intersections $(x_0,y_0), (x_0,-y_0)$, the center $(d/2,0)$ of the ellipse, the radii $r_1$ and $r_2$\label{f:not}}
\end{figure}
The  cutting depth $z_0$ of the circle and the cutting depth $z_1$ for the ellipse are
\begin{eqnarray*} 
z_0  & = & w\cdot \frac{2d-2+w}{2d}\\
z_1 &=&
 1- z_0  +w/2
= \frac{(2-w)(d+w)}{2d}\ .
 \end{eqnarray*}
The radii of the ellipse are given by
\begin{eqnarray*} 
r_1 & = & \frac{d+w}2\ , \\
r_2 & = & \frac12 \sqrt{(d+w)^2 - d^2}\  =\  \frac12  \sqrt{(2d+w)w}\ .
\end{eqnarray*}
The other parameters are 
\begin{eqnarray*} 
x_0 & = & 1- z_0(w,d) \ = \ 1- w\cdot \frac{2d-2+w}{2d} \\
y_0 & = & \sqrt{r^2 - x_0^2} \ , 
\end{eqnarray*}
which are the solutions to the equations
\begin{eqnarray*} 
x_0^2 + y_0^2& = & 1 \\
(d-x_0)^2 + y_0^2 & = & (d-1+w)^2 \\
r_1^2 + \frac14 d^2&=& \frac14(d+w)^2  \ .
\end{eqnarray*}

%
This implies for $f(w,d):= |E_{\leq w} \cap U|$ (time shift $w$ and receiver in distance $d$):
\begin{eqnarray*}
f(w,d) &=& s(1,1,z_0) + s(r_1,r_2, z_1)\\
&=&  g(z_0)+ r_1r_2 \ g\left( \frac{z_1}{r_1}\right)\\
&=& g\left( \frac{w(2d+w-2)}{2d}\right)+\frac14 (d+w) \sqrt{w(2d+w)}\ 
g\left( \frac{2-w}{d }\right)
\end{eqnarray*}\label{AI-proof}
Define the derivatives $f'$ and $f''$ with respect to the first parameter:
\begin{eqnarray}
	 f'(x,y)&  := &\quad\frac{\mathrm{d} f(x,y)}{\mathrm{d} \ x}  \label{eq:neuns} \\
	f''(x,y)& := &\quad\frac{\mathrm{d}^2 f(x,y)}{\mathrm{d}^2 \ x} \label{d:fss}
	\end{eqnarray}
The following observations of the derivatives are useful later on.
\begin{lemma}
For $x\in [0,2]$, $y\geq 1$:
\begin{eqnarray}
	 f'(x,y) & > & 0 \label{eq:neun} \\
	\sqrt{x} \ \ < \ \ f(x,y) & < &  \frac{7}3 \sqrt{x}  \label{eq:fff}
	\label{eq:zehn} \\
\hbox{For $y\geq 2$:}\quad 	\frac32 \sqrt{x} \ \ < \ \ f(x,y)  \label{eq:ffff}\\
\hbox{For $y>1$, $x\leq 2$:}\quad	\frac{f(x,y)}{f(x/2,y)} & > & \frac{7}{5} 
\label{eq:droelf} 
\end{eqnarray}
\label{la1}
\end{lemma}
\begin{proof}
Inequality~(\ref{eq:neun}) follows by the definition of $f(x,y)$. The other inequalities have been verified by computer generated function tables using interval arithmetics, see also the plots in Figs.~\ref{f1}-\ref{f3} to get an intuition.

\begin{figure}[ht]
\centerline{\includegraphics[width=.5\linewidth]{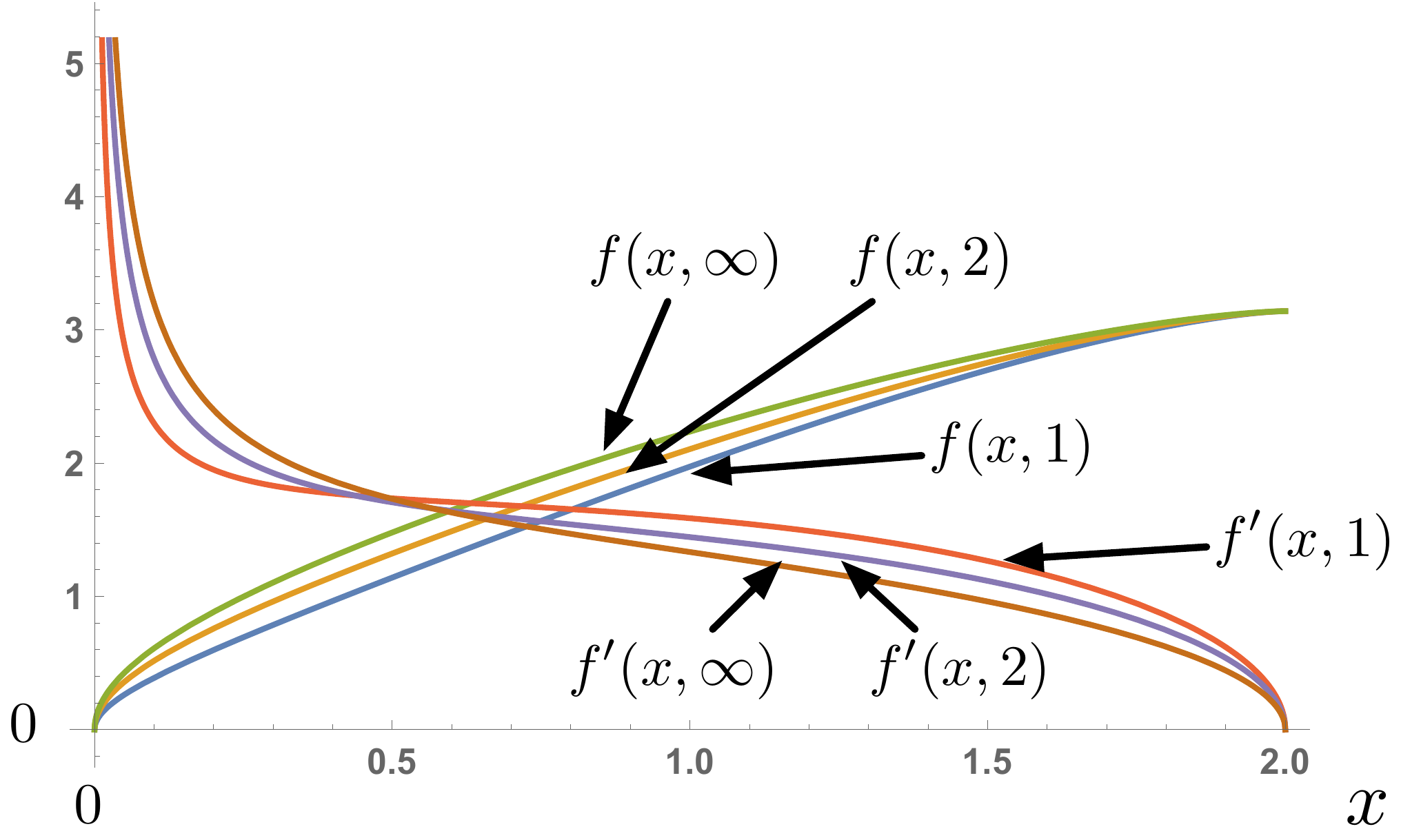}}
\caption{Plot of functions $f(x,y)$ and $f'(x,y)$ for exemplary values $y\in\{1,2,\infty\}$\label{f1}}
\end{figure}

In the automated proofs we only use rational numbers and only test if for two numbers $a<b$. The automated proofs check whether for some function $h$ the inequality 
$h(x) < c$ holds for $x\in [a,b]$, where $h(x)$ is finite in $[a,b]$. For this we divide the interval $[a,b]$ in smaller intervals $[a_0,a_1], [a_1,a_2], \ldots, [a_{n-1},b]$. Then, for an interval we calculate $f([a_j,a_{j+1}]) = [s_j,\ell_{j}]$ using interval arithmetics \cite{hickey2001interval}. So, it holds that $s_j \leq f(x) \leq \ell_{j+1}$ for all $x \in [a_j,a_{j+1}]$. Then, we check whether $\ell_j < c$ for all $i$. If this is not the case, we increase the number of intervals $n$ to have smaller intervals. By basic analysis it follows that a set of small enough intervals must exist that deliver a positive result if $h(x)<c$.

Analogously, we  generate proofs for two-dimensional cases to check that $f(x,y) < c$ holds for $x \in [a,b], y\in [a',b']$, where we use input intervals $[a_j,a_{j+1}] \times [a'_k, a'_{k+1}]$.

However, a direct application is not possible before we have not resolved the following two  problems: we have an open interval $y\in [1,\infty)$ and some comparisons involve functions on both side of the inequality. 
 
For the first problem we substitute $z=\frac1{y}$ and consider functions over the interval $z \in [0,1]$. At $z=0$ we have to show that the discontinuity is removable.  This way we prove Inequality~(\ref{eq:droelf}) noting that $\lim_{y\rightarrow \infty} \frac{f(2,y)}{f(1,y)} = \sqrt{2} \geq \frac75$.

For the second problem we divide by one side and get a removable discontinuity.
In order to remove the discontinuity at $x=0$ for $f(x,y)/\sqrt{x}$, we observe that
$$\lim_{y\rightarrow \infty} f(x,y) = \frac{1}{3} (x+1) \sqrt{(2-x) x}+\arccos (1-x)
$$
such that we get a removable discontinuity for $\frac{f(x,y)}{\sqrt{x}}$ for $y\rightarrow \infty$. Using the substitution $z=1/y$ the automatic selection of intervals for interval arithmetics produces the proofs for inequalites~(\ref{eq:fff}) and (\ref{eq:ffff}).
\end{proof}

\begin{figure}[ht]
\centerline{\includegraphics[width=.5\linewidth]{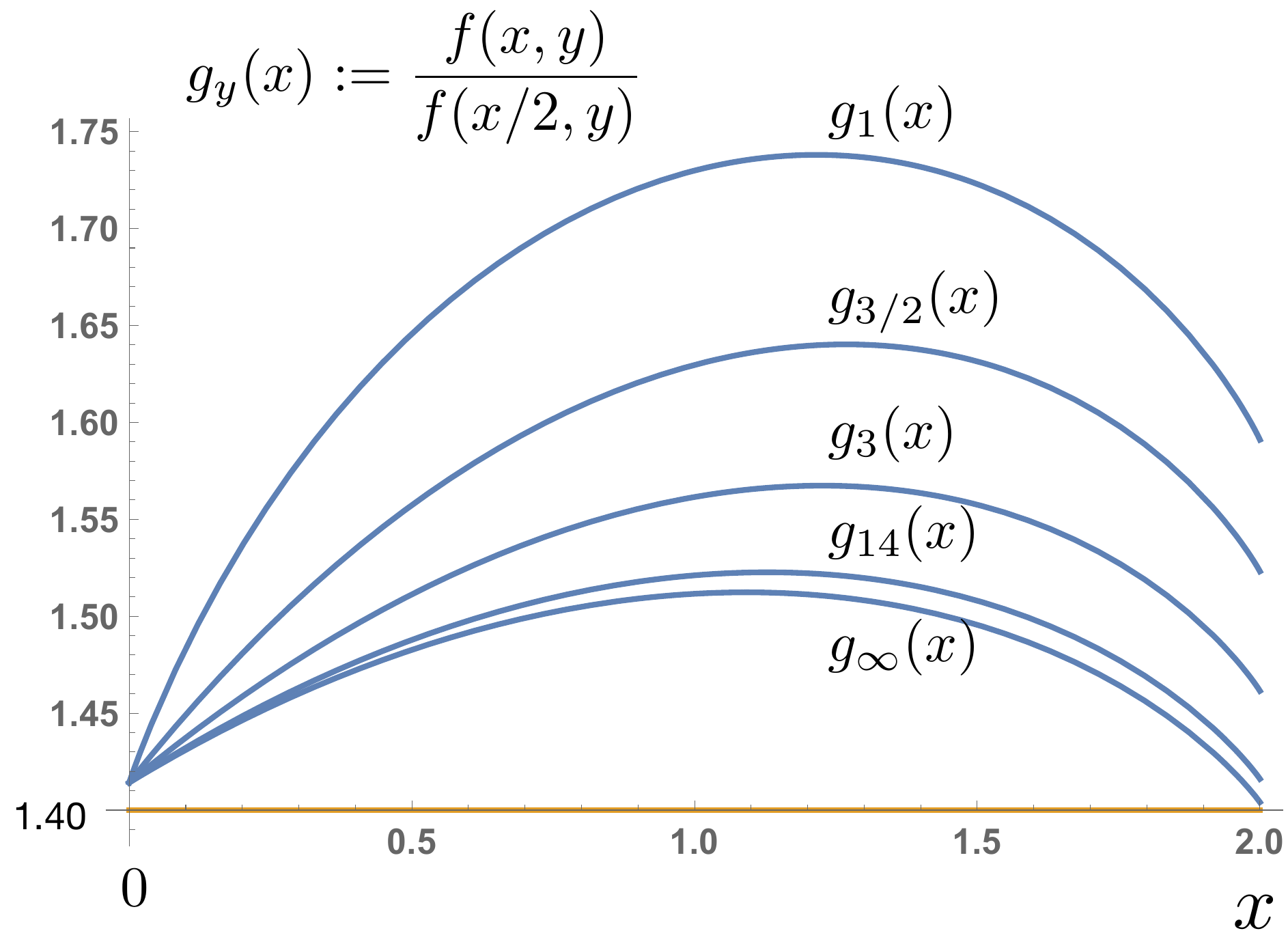}}
\caption{The function $f(x,y)/f(x/2,y)$ for  $y\in\{1,3/2,3,14,\infty\}$ relevant for Inequality~(\ref{eq:droelf}) \label{f4}}
\end{figure}

For (\ref{eq:neuns}), (\ref{eq:fff}), (\ref{eq:ffff})  $f$, $g$, the first and second derivatives of $f$ are the following. 
\begin{eqnarray*}
f(x,y) & = & g\left( \frac{x(x+2y-2)}{2y} \right)+  \frac14 (x+y) \sqrt{x(x+2y)}  \ g\left(\frac{2-x}{y}\right)  \\
\frac{\mathrm{d} \ g(x)}{\mathrm{d} \  x} & = & \sqrt{1-(1-x)^2} \\
f'(x,y)&:=& \frac{df(x,y)}{dy} \\&=&
 g\left(\frac{2-x}{y}\right) \frac{2 x^2+4 x y+y^2}{4 \sqrt{x (x+2 y)}} 
 +
 \ \frac{(x+y-2)\sqrt{x(2-x)(x+2 y-2)(x+2 y) }}{2 y^2} \\
 f''(x,y)&:=& \frac{d^2f(x,y)}{\mathrm{d}^2 \ y} =T_1(x,y)+ T_2(x,y) + T_3(x,y) \ ,
 \end{eqnarray*}
 where we have the following terms.
 \begin{eqnarray*}
 T_1(x,y) 
  & = & \frac{
  -3 x^4-2 x^3 (6 y-7)-2 x^2 \left(7 y^2-21 y+10\right)
  -4 x \left(y^3-8 y^2+10 y-2\right)+4 y \left(y^2-3 y+2 \right)
  }{2 y^2 \sqrt{(2-x) x (x+2 y-2) (x+2 y)}} \ \\[2ex]
  T_2(x,y)& = & -\frac{ \sqrt{(2-x) (x+2 y-2)} \left(2 x^2+4 x y+y^2\right)}{2 y^2 \sqrt{x (x+2 y)}}\\
T_3(x,y)& = &\frac{(x+y) \left(2 x^2+4 x y-y^2\right) }{4 (x (x+2 y))^{3/2}}  g\left(\frac{2-x}{y}\right)
\end{eqnarray*}

We can prove via interval arithmetics for $x \in [0,2], y\geq 1$ the following inequalities, see also Figures \ref{ff1}, \ref{ff2}, \ref{ff3}, \ref{ff4}.
\begin{lemma}
The following inequalities hold for $y\geq1$, $x \in [0,2]$:

\begin{eqnarray}
-2\  & \leq &\ \   T_1(x,y) \sqrt{x(2-x)} \nonumber  \\[-3.3ex]&&\hspace{25ex}\leq\   2 \label{T1}\\[1ex]
-3\  & \leq &\ \   T_2(x,y) \sqrt{x} \nonumber     \\[-3.3ex]&&\hspace{25ex}\leq\  0 \ \label{T2}\\[1ex]
-1 \   & \leq &\ \   T_3(x,y) x^{3/2}   \nonumber   \\[-3.3ex]&&\hspace{25ex}\leq\  1 \label{T3} \\[1ex]
1 \    &\leq &\quad  \frac{g(x)}{x^{3/2}}  \nonumber\\[-3.9ex]&&\hspace{25ex}\leq\  2  \label{g32} \ .
\end{eqnarray}

For  $x\in \left(0, \frac1{100}\right], y\geq 1$
\begin{equation}
 T_3(x,y) x^{3/2} \   \leq \  -\frac15 \ . \label{T33215}
 \end{equation}
where discontinuities at $x=0$ and $x=2$ are removed. 

For $x\in \left[\frac1{100}, 2-\frac1{100}\right]$, $y\geq 1$
\begin{equation} \label{fss18}
 f''(x,y)\  \leq\  -\frac18\ .
 \end{equation}
\end{lemma}
\begin{proof}
Possible discontinuities at $x=0$ and $x=2$ of the functions $T_1(x,y) \sqrt{x(2-x)}$, $T_2(x,y) \sqrt{x}$, and $T_3(x,y) x^{3/2}$ can  trivially be removed by algebraic transformation.
The function $\frac{g(x)}{x^{3/2}}$ is also finite with a removable discontinuity at $x=0$, since
 $$\lim_{x\rightarrow 0} \frac{g(x)}{x^{3/2}} = \frac43 \sqrt2 \ ,$$
 which can be seen by the Taylor series.
 For the automated proofs we again replace $z= \frac1y$ and prove the claims for the interval $z\in [0,1]$ with interval arithmetics.
 
 Since $f''$ is finite over this domain, we can apply an automated proof with interval arithmetics for $f''(x,1/z) < -\frac18$ for $x \in [1/100,1-1/100]$ , removing the discontinuity for $z=0$.
\end{proof}

\begin{figure}[ht]
\centerline{\includegraphics[width=.5\linewidth]{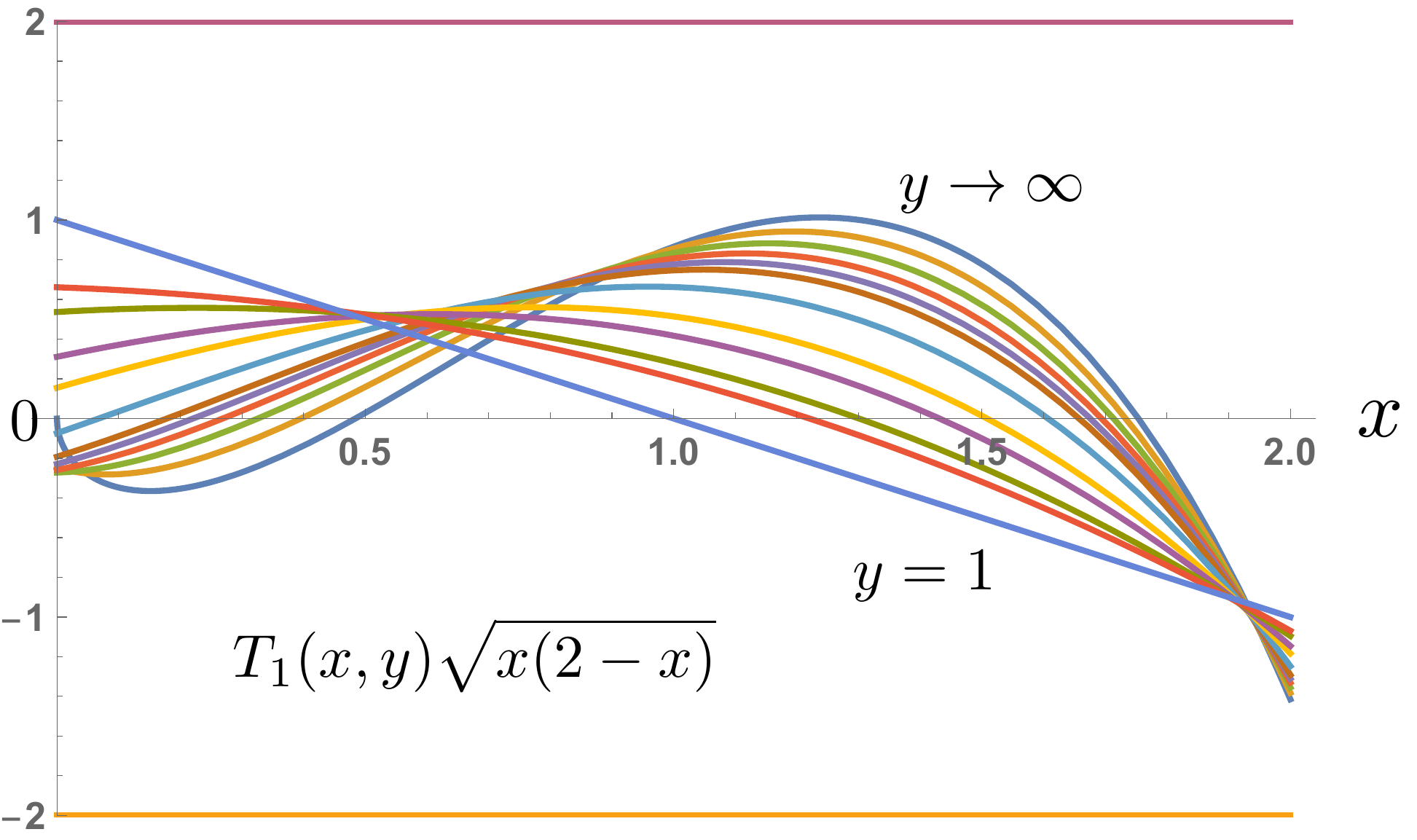}}
\caption{The function $T_1(x,y) \sqrt{x(2-x)}$ for 
$y\in \{1, 1.1, 1.2, 1.3, 1.4, 1.5, 1.8, 2.5, 3.2, 5, 7, \infty\}$, relevant for Inequality~(\ref{T1}) \label{ff1}}
\end{figure}

\begin{figure}[ht]
\centerline{\includegraphics[width=.5\linewidth]{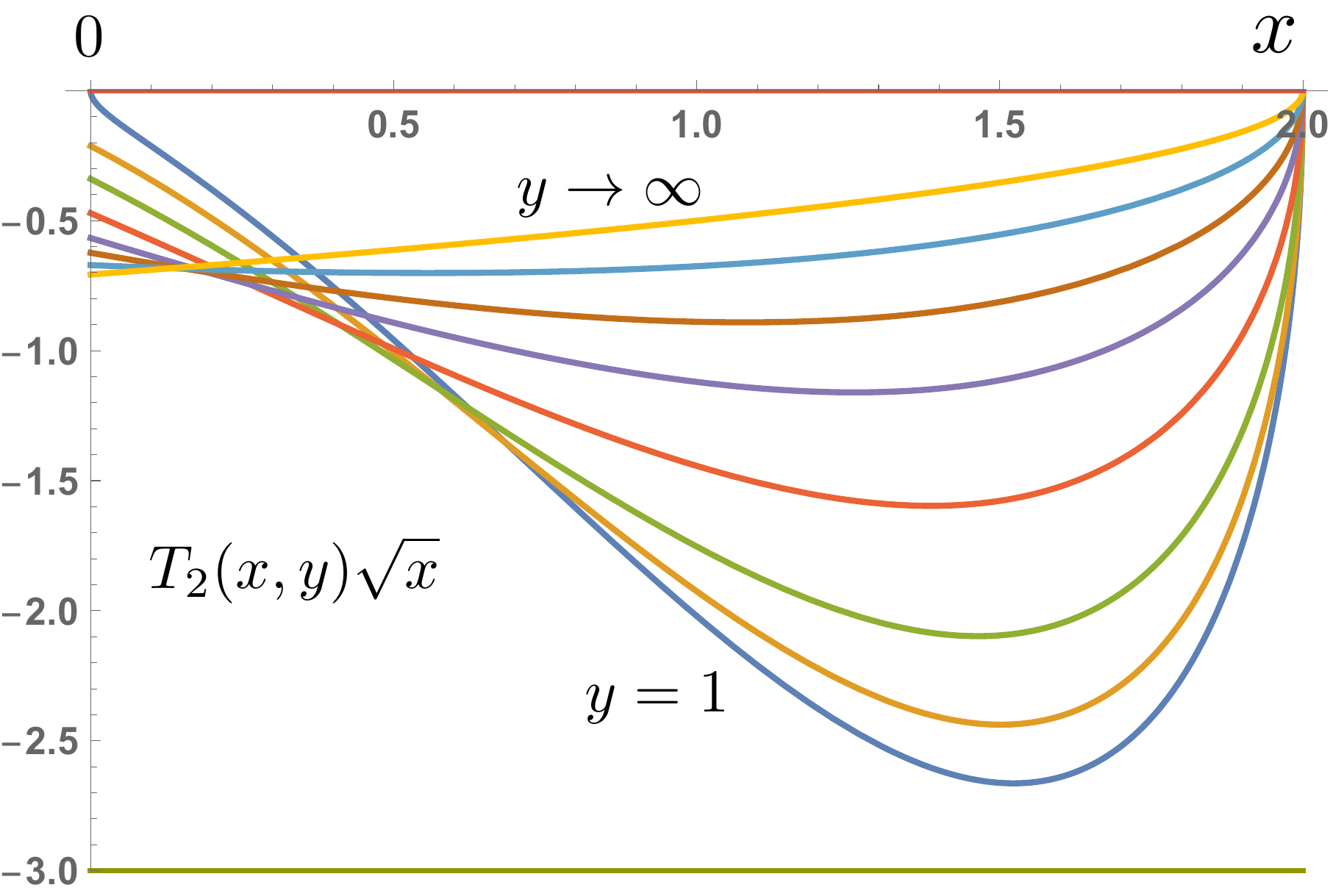}}
\caption{The function $T_2(x,y) \sqrt{x}$ for 
$y\in \{1, 1.1, 1.3, 1.8, 2.8, 4.5, 10,\infty\}$, relevant for Inequality~(\ref{T2}) \label{ff2}}\end{figure}
\begin{figure}[ht]
\centerline{\includegraphics[width=.5\linewidth]{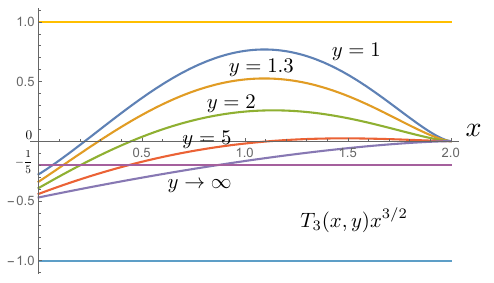}}
\caption{The function $T_3(x,y) x^{3/2}$ for 
$y\in \{1, 1.3, 2, 5,\infty\}$
, relevant for Inequality~(\ref{T3}) \label{ff3}}
\end{figure}

\begin{figure}[ht]
\centerline{\includegraphics[width=.5\linewidth]{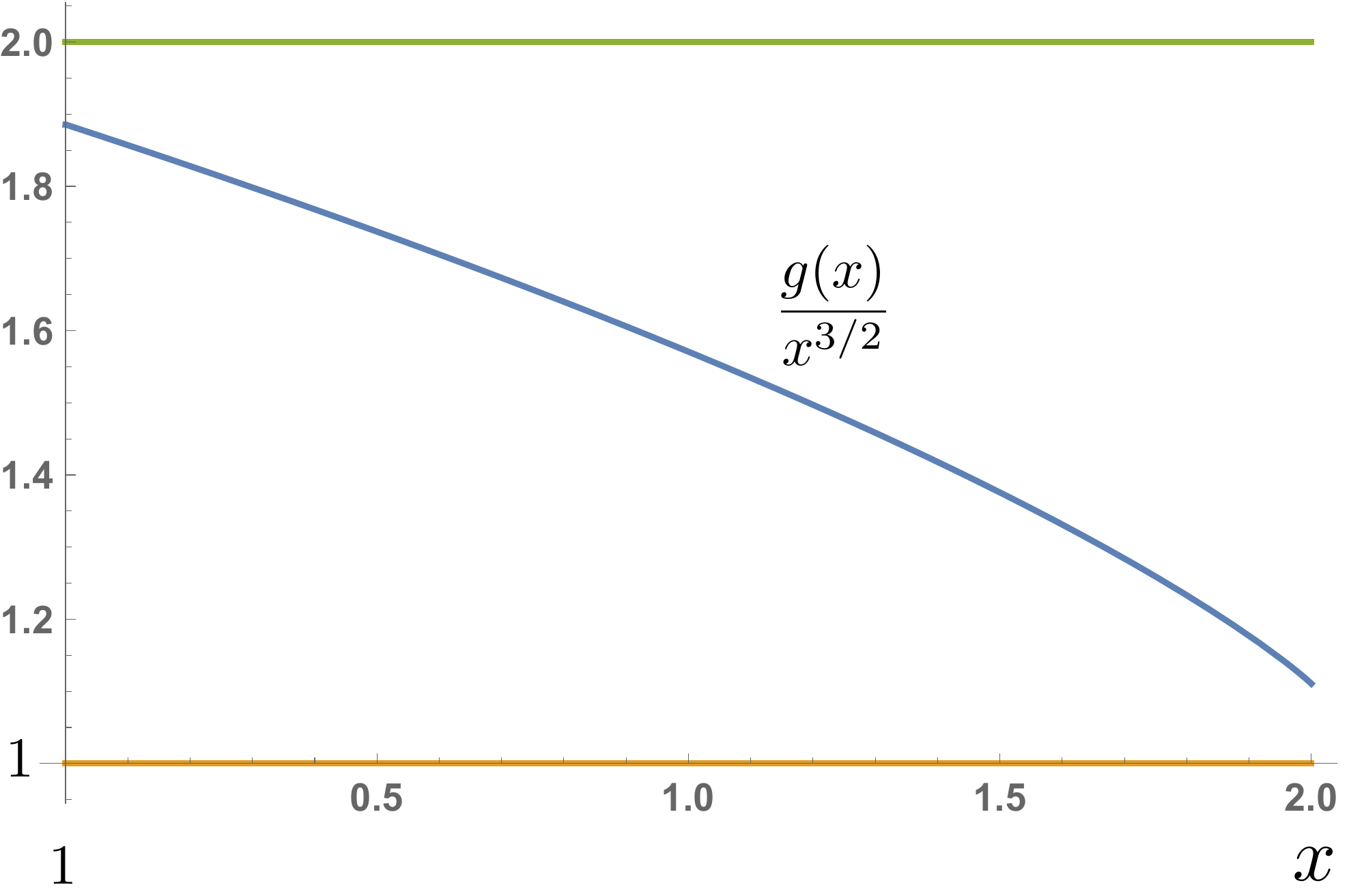}}
\caption{The function $g(x)/x^{3/2}$, relevant for Inequality~(\ref{g32}) \label{ff4}}
\end{figure}

 \begin{lemma} \label{l:fss}
For all $y>1$, $x\in (0,2)$
$$f''(x,y)<-\frac14\ .$$
For $x\in \left[ 2-\frac1{100},2\right)$, $y\geq 1$
\begin{equation}
 f''(x,y)\  \leq\  -199\ .
\end{equation}
\end{lemma}
\begin{proof}
The second derivative $f''(x)$ tends to $-\infty$ at the borders of $x\in [0,2]$. 
We have  proved (\ref{T33215}) via interval arithmetics for $x\leq \frac{1}{100}$:  
$x^{3/2} T_3(x,y)  \leq -1/5$. Then, it follows for $x \in [0, \frac1{100}]$:
\begin{figure}[ht]
\centerline{\includegraphics[width=.5\linewidth]{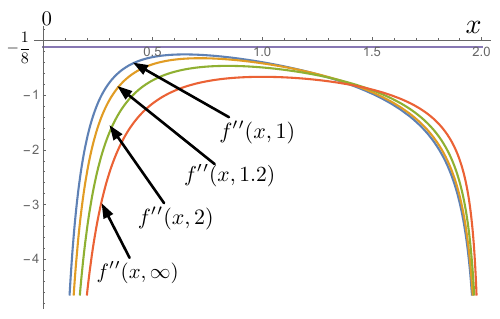}}
\caption{The function $f''(x,y)$ for  $y\in\{1,1.2,2,\infty\}$ of Lemma~\ref{l:fss} \label{f:fss}}
\end{figure}
\begin{eqnarray*}
 f''(x,y) & \leq  &  T_3(x,y)+ \frac{2}{\sqrt{x(2-x)} } \\
 & \leq  & -\frac{1}{5 x^{3/2}} + \frac{2\sqrt{x}}{\sqrt{(2-x)} } \\
 & =  & \frac{1}{\sqrt{x}} \left(-\frac{1}{5 x} + \frac{2 x}{\sqrt{(2-x)} }\right) \\
 & <  & -199
\end{eqnarray*}

For $x\geq 2-\frac1{100}$ we get
\begin{eqnarray*}
 f''(x,y) & \leq  &  T_1(x,y) + T_2(x,y) + T_3(x,y) \\
 & \leq  &  -\frac12 \sqrt{x(2-x)} + x^{3/2} \\
 & \leq  &  -\frac{\sqrt{100}}{2 \sqrt{2-\frac{1}{100}}} + (2-1/100)^{3/2} \\
 & < & -7/5
\end{eqnarray*}
\end{proof}
These lemmas imply:
\begin{lemma}\label{l:fsss}
For $x\in (0,2), y\geq 1$:
$$f''(x,y)\leq -\frac18\ .$$
\end{lemma}
 
\begin{figure}[ht]
\centerline{\includegraphics[width=.5\linewidth]{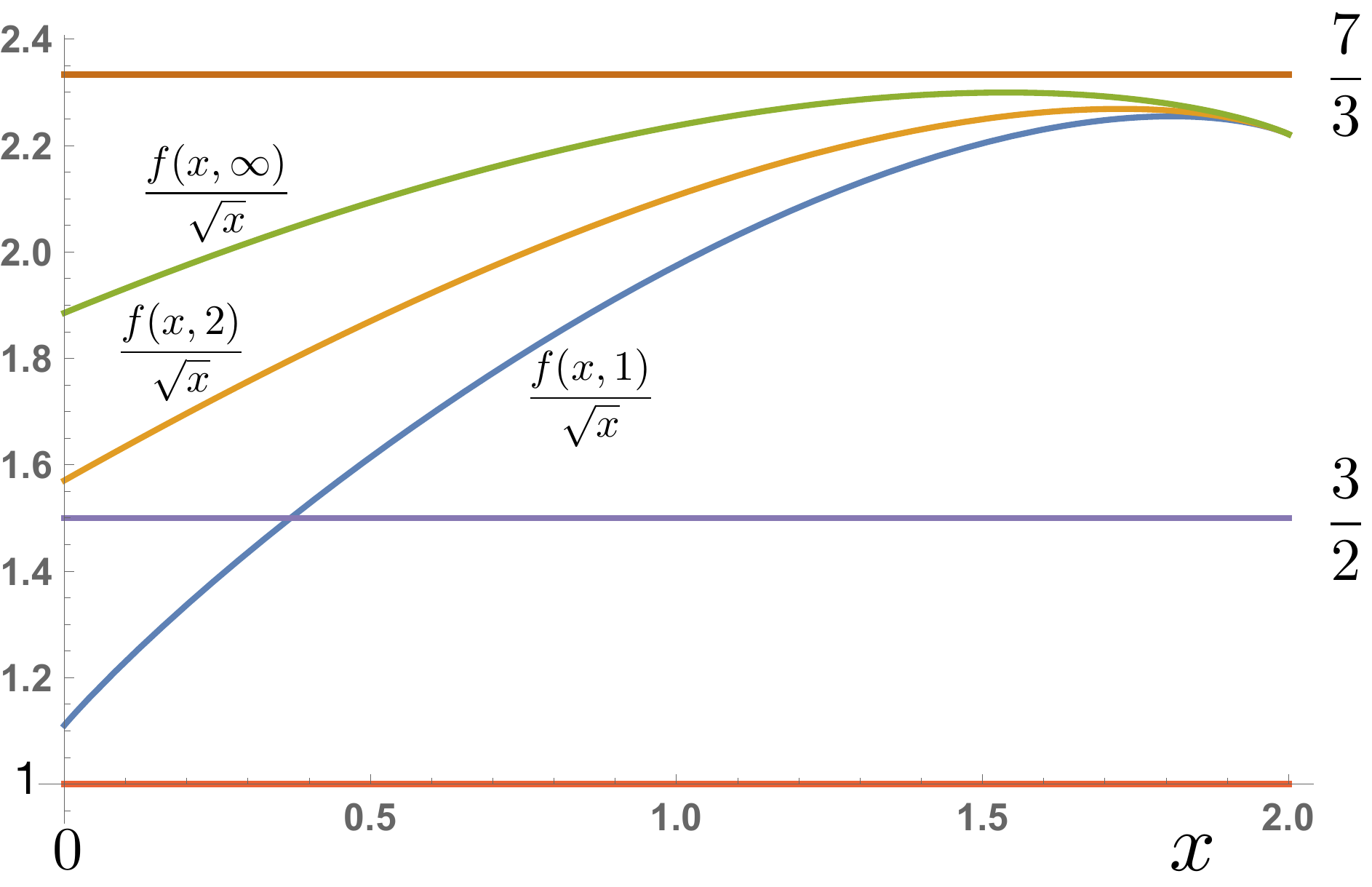}}
\caption{The function $f(x,y)/\sqrt{x}$ for  $y\in\{1,2,\infty\}$ relevant for Inequalites~(\ref{eq:fff}) and (\ref{eq:ffff}) \label{f3}}
\end{figure}

Using $f$ we can describe $h_{\infty}$ as follows.

\begin{lemma}For $w\in [0,2]$:
$$  h_{d, \lambda}(w) = \int_{x = 0}^{w} e^{i 2\pi x/ \lambda} f'(x,d)
  \ \mathrm{d}x\ .$$  
\end{lemma}
\begin{proof}
From the above considerations we know that the area of $E_{\leq w} \cap U$ is described by
$ f(w, d)\ .$
Hence, the differential area with points with phase $w$ is given by $$\frac{\mathrm{d}\ f(w,d)}{\mathrm{d}\ w} =  f'(w,d)\ .$$
\end{proof}

For the orientation of $h_{d, \lambda}(w)$ we get the following Lemma.
\begin{lemma} For all $w \geq 0$:  
$$ \arg(h_{d, \lambda}(w)) \in (0,\pi)\ .$$ 
\end{lemma}
\begin{proof}
We extend $f'(x,y)$ as $f'(x,y)= 0$ for $y\geq 2$.
We consider two consecutive intervals of distance $\lambda$:
\begin{eqnarray*}
h_{d, \lambda}(w)&=&\sum_{j=0}^{\lceil 2/\lambda \rceil} \int_{x = 0}^{\lambda} e^{i 2\pi x/ \lambda} f'\left(x+j \lambda,{d}\right)
  \ \mathrm{d}x\ \\
   &=& \sum_{j=0}^{\lceil 2/\lambda \rceil} \int_{x = 0}^{\lambda/2}e^{i 2\pi x/ \lambda}\left(f'\left({x+j \lambda},{d}\right) 
    f'\left({x+j \lambda + \frac12 \lambda},{d}
   \right)\right)
  \ \mathrm{d}x\ \\
    &=& \int_{x = 0}^{\lambda/2} e^{i 2\pi x/ \lambda} \sum_{j=0}^{\lceil 2/\lambda \rceil}  \left(f'\left({x+j \lambda},{d}\right)-
    f'\left({x+j \lambda + \frac12 \lambda},{d}
   \right)\right)
  \ \mathrm{d}x\ \\
\end{eqnarray*}  
Since $f''(x,y) \leq -\frac18$ we have
%
%
$$f'\left({x+j \lambda},{d}\right)-f'\left({x+j \lambda + \frac12 \lambda},{d}
   \right)\geq \frac1{16} \lambda$$
   
Because every difference is positive and it follows that $\arg(h_{d, \lambda}(w)) \in (0,\pi)$.
\end{proof}

From this consideration one can only imply that $\Im(h_{d, \lambda}(w)) = \Omega(\lambda)$  for $w\geq \lambda$. To achieve a better bound we concentrate on the first term of the sum.

\begin{proposition} For $x\leq \lambda/2$:
$$f'\left({x},{d}\right) \ \geq \ \frac{7}{5} f'\left({x + \frac12 \lambda},{d}
   \right)$$
\end{proposition} 
\begin{proof}
Since $x\leq \lambda/2$ we have $2x \leq x+\lambda$ and by  Lemma~\ref{la1} and the monotony of $f'$ it follows:
   $$f'({x},{d}) \ \geq \ \frac{7}{5} f'({2 x},{d})\ \geq \ \frac{7}{5}\  f'\left({x + \frac12 \lambda},{d}
   \right)\ .$$
\end{proof}
\begin{proposition} For $w \geq \lambda/2$:
$$ \Im(h_{d, \lambda}(w)) \geq  \int_{x = 0}^{\lambda} \sin(2\pi i\ x/ \lambda) \ f'\left({x},{d}\right)
  \ \mathrm{d}x\ .$$
\end{proposition} 
\begin{proof}
Because $f''(x,y)<0$ the other sum terms for $j>0$ may only increase the imaginary terms.
\end{proof}
\begin{proposition} 
$$ \int_{x = 0}^{\lambda/2} \sin(2\pi i\ x/ \lambda) \ f'\left({x},{d}\right)
  \ \mathrm{d}x\ \geq  
    \frac27
   \int_{x = 0}^{\lambda} \sin(2\pi i\ x/ \lambda) \ f'\left({x},{d}\right)
  \ \mathrm{d}x\ $$
\end{proposition}
\begin{proof} For $x\leq \lambda$ by the claim above we have:
$$f'\left({x},{d}\right) \geq \frac75 \ 
f'({x+\lambda/2},{d})$$
Hence 
$$f'(x,d) -f'\left(x+\lambda/2,d\right) \ \ \geq \ \ \frac27 \ f'(x,d) \ .$$
\end{proof}

\begin{proposition} For $d\geq 2$
$$ \int_{x = 0}^{\lambda/2} \sin(2\pi i\ x/ \lambda) \ f'\left({x},{d}\right)
  \ \mathrm{d}x\ \ \geq  \  \frac1{16\sqrt2} \sqrt{\lambda}
     $$ 
\end{proposition}
\begin{proof}
We use that $\frac32 \sqrt{x}< f(x,y) < \frac73 \sqrt{x}$ from (\ref{q73}) and that $f'(x,y)$ decreases with $x$ (Lemma~\ref{l:fsss}).
\begin{eqnarray*}
\int_{x = 0}^{\lambda/2} \sin(2\pi i\ x/ \lambda) \ f'\left(x,d\right)
  \ \mathrm{d}x\ 
   & \geq & 
  \int_{x = \lambda/4}^{\frac34 \lambda} \sin(2\pi i\ x/ \lambda) \ f'\left(x,d\right)
  \ \mathrm{d}x\ \\
  & \geq & 
  \int_{x = \lambda/4}^{\frac34 \lambda} \frac1{\sqrt2} \ f'\left(x,d\right)
  \ \mathrm{d}x\ \\
   & \geq & 
  \int_{x = \lambda/4}^{\frac34 \lambda} \frac1{\sqrt2} \ f'\left(x,d\right)
  \ \mathrm{d}x\ \\
   & \geq & 
    \frac{1}{2\sqrt2} \ \left( f\left({\frac34 \lambda},d\right)-
  f\left({\frac14 \lambda},d\right)\right)\\
  & \geq &  \frac{1}{2\sqrt2} \left(\frac32 \sqrt{\frac34 \lambda}- 
     \frac73 \sqrt{\frac14 \lambda}\right)\\
  & \geq &  \frac{\frac32\sqrt3-\frac73}{4\sqrt2}\ 
 \lambda^{1/2}  \label{q73} \\
  & \geq &  \frac{1}{16\sqrt2}\ 
 \lambda^{1/2}  \end{eqnarray*}
\end{proof}

This implies the following Lemma:
\begin{lemma} \label{lA10}For $w \geq \lambda/2$, $d \geq 2r$:
$$
 \Im(h_{d, \lambda}(w)))  \geq \frac{1}{56\sqrt2} \sqrt{\lambda} \ .$$
 \end{lemma}
 \begin{proof}
 Follows by the combining the above claims.
 \end{proof}
 
We now estimate for $d\geq c_f \lambda +1$:
$$ u_{d,\lambda}(w) :=\iint\displaylimits_{(x,y) \in E_{\leq w} \cap U} 
\frac{e^{2\pi
i \Delta_d(x,y)/\lambda}}{\sqrt{(x-d)^2+y^2}} \  \mathrm{d} x \ \mathrm{d} y\ ,$$
which is of interest since $u_{d,\lambda}(2) = s(d,\lambda,1)$.
Note that $\Delta_d(p):= |\!|p|\!|_2 + |\!|p-q|\!|_2- d\ .$ We change coordinates from $(x,y)$ to $(z, \ell)$, where $z$ is the $x$-coordinate of the Ellipse $E_{w}$ and $\ell$ is the distance from the point to the receiver. Note that
$$\left(
\begin{matrix}
x(\ell,z) \\
y(\ell,z)
\end{matrix}\right) = 
\left(
\begin{matrix}\displaystyle
\frac{2 d^2-2 d \ell+2 d z-2 \ell z+z^2}{2 d}\\[2ex]
\displaystyle
\frac{\sqrt{z} \sqrt{2 \ell-z} \sqrt{4 d^2-4 d \ell+4 d z-2 \ell z+z^2}}{2 d}\end{matrix}\right) 
$$

\begin{lemma} For $0 \leq w \leq 2$:
 \begin{eqnarray*}
 u_{d,\lambda}(w)
 &=& \int\displaylimits_{z=0}^{w} 
 \int\displaylimits_{\ell=d-1+z}^{d+z/2}  
  2\frac{e^{2\pi i z/\lambda}}{\ell} 
  |\det(\phi(\ell,z))| \mathrm{d}\ell \  \mathrm{d}z\  
\end{eqnarray*}
where for $\ell\geq z$:
$$|\det(\phi(\ell,z))| = \frac{2 \ell (d-\ell+z)}{\sqrt{z} \sqrt{2 \ell-z} \sqrt{(2 d+z) (2 d-2 \ell+z)}}$$
\end{lemma}
where $\phi(\ell,z)$ is the Jacobian of 
$\left(
\begin{matrix}
x(\ell,z) \\
y(\ell,z)
\end{matrix}\right)$.  

\begin{proposition}
$$
 f'(z,d) = 
\int\displaylimits_{\ell=d-r+z}^{d+z/2}   
 2 |\det(\phi(\ell,z))| \ \mathrm{d}\ell  $$
\end{proposition}

Define 
$$t(d,z, \ell) := \int\displaylimits_{\ell=d-1+z}^{d+z/2}  
  \frac{2}{\ell} 
  |\det(\phi(\ell,z))| \mathrm{d}\ell  \ .$$
\begin{lemma} \label{A13}
 For  $d > 1$ and $\ell \in [d-1,d+1]$ for all $w \in [0,2]$:
$$\frac{f'(z,d)}{d+1} \ \ \leq  \ \ t(d,z, \ell) 
 \ \ \leq \ \ 
\frac{f'(z,d)}{d+1} \left(1+ \frac{2}{d-1} \right)$$
\end{lemma}
\begin{proof}
Note that $0 < d-1 \leq \ell \leq d+1$. Hence,
$$ \frac{1}{d+1} \ \ \leq \ \ \frac{1}{\ell} \leq \frac{1}{d-1} = \frac1{d+1} 
\left(1+\frac{2}{d-1}\right)$$
\end{proof}

We choose $d \geq 15$ and get $\frac{2}{d-1} \leq \frac17$.

\begin{lemma} For $w \geq \lambda/2$, $d \geq 15$:
$$
 \Im(s(d,\lambda,1)) \geq \frac{1}{112\sqrt2 (d+1)}\sqrt{\lambda} \ .$$
\end{lemma}
\begin{proof}
We consider two consecutive intervals of distance $\lambda$:
\begin{eqnarray*}
u_{d,\lambda}(w)  &=&\sum_{j=0}^{\lceil 2/\lambda \rceil} 
 \int\displaylimits_{z = 0}^{\lambda}
 \int\displaylimits_{\ell=d-1+z+j \lambda}^{d+(z+j\lambda)/2}  
  2\frac{e^{2\pi i z/\lambda}}{\ell} 
  |\det(\phi(\ell,z))| \ \mathrm{d}\ell \ \mathrm{d}z\ 
   \\
&=&  \sum_{j=0}^{\lceil 2/\lambda \rceil} 
 \int\displaylimits_{z = 0}^{\lambda/2}\frac{e^{2\pi i z/\lambda}}{\ell}
 \left( t(d,z+j \lambda,\ell) -t(d,z+j\lambda+\lambda/2) \right)
 \  \mathrm{d}z
 \end{eqnarray*} 
Now for $z\le \lambda/2$, $d\geq 15$:
\begin{eqnarray}
t(d,z+\lambda/2, \ell(x,y)) &\leq& \frac87\frac{1}{d+1} 
f'(z+\lambda/2,d) \label{ineq89}\\
&\leq& \frac87 \cdot\frac57 \cdot \frac1{d+1}  f'(z/2 + \lambda/4,d)\nonumber \\
&\leq& \frac87 \cdot\frac57 \cdot \frac1{d+1}  f'(z,d) \label{eq5787}\\
&\leq&
\frac{40}{49} t(d,z, \ell(x,y)) \nonumber
\end{eqnarray} 
The first line (\ref{ineq89}) follows from Lemma~\ref{A13}.
Line~(\ref{eq5787}) follows from $z/2 + \lambda/4 \geq  z$ and that $f'$ is monotone decreasing because $f''(x,y)< -\frac18$ by Lemma~\ref{l:fsss}.

Hence $t(d,z, \ell(x,y)) - t(d,z+\lambda/2, \ell(x,y)) \geq \frac{9}{40} t(d,z, \ell(x,y)$ which we use in line~\ref{l940}:
\begin{eqnarray}
 \Im(s(d,\lambda,1))  &=&  \Im(u_{d,\lambda}(2))  \\
 &=&  \sum_{j=0}^{\lceil 2/\lambda \rceil} 
 \int\displaylimits_{z = 0}^{\lambda}
  \frac{\sin(2\pi i\ z/ \lambda)}{\ell}  t(d,z+j\lambda, \ell(x,y)) 
  \ \mathrm{d}\ell\ \ \mathrm{d}z \nonumber\\
  & \geq & 
  \sum_{j=0}^{\lceil 2/\lambda \rceil} 
 \int\displaylimits_{z = 0}^{\lambda/2}
 \frac{9}{40}  \frac{\sin(2\pi i\ z/ \lambda)}{\ell}  
 t(r,d,z+j\lambda, \ell(x,y)) 
  \ \mathrm{d}\ell\ \mathrm{d}z  \label{l940}\\
  & \geq & 
  \sum_{j=0}^{\lceil 2/\lambda \rceil} 
  \int\displaylimits_{z = 0}^{\lambda/2}
  \frac{9}{40} \frac{\sin(2\pi i\ z/ \lambda)}{d+1} f'(z+j \lambda,d) \ \mathrm{d}z  \nonumber\\
   & \geq &   \frac{9}{40} \frac{1}{d+1}
  \sum_{j=0}^{\lceil 2/\lambda \rceil} 
  \int\displaylimits_{z = 0}^{\lambda}
 \sin(2\pi i\ z/ \lambda) f'(z+j \lambda,d) \ \mathrm{d}z  \nonumber\\
  & = &  \frac{9}{40}\frac{1}{d+1}  \Im(h_{d,\lambda}(2))\label{hdef}\\
  & \geq & 
  \frac{1}{56\sqrt2}\frac{9}{40}\frac{1}{d+1} \sqrt{\lambda}   \nonumber\\
  & \geq & 
  \frac{9}{2,240\sqrt2}\frac{ \sqrt{\lambda}}{d+1} \nonumber
%
 \end{eqnarray} 
\end{proof}
In line~(\ref{hdef}) we use the definition of $h_{d,\lambda}(2)$ and then apply Lemma~\ref{lA10}.
\end{proof}

\begin{lemma}For $d>15r$:
$$\Exp{\Im[\hbox{\sl rx}]}  \geq 
 \frac{9}{4,480\pi\sqrt2} \frac{m\sqrt{\lambda}}{d\sqrt{r}}
\ .$$
\end{lemma}
\begin{proof}$
\Exp{\Im[\hbox{\sl rx}]}  \geq 
\frac{m}{2\pi r} 
\frac{9}{2,240\sqrt2} \frac{\sqrt{\lambda/r}}{d/r+1} \ \  \geq \ \ 
 \frac{9}{4,480\pi\sqrt2} \frac{m\sqrt{\lambda}}{d\sqrt{r}}
$.
\end{proof}

We apply the Hoeffding bound (Theorem 2 of \cite{hoeffding1963probability}), which states for $n$ independently distributed random variables $X_j\in \REAL$  strictly bounded by the intervals $[a_j, b_j]$:
\begin{equation*} \Prob{\left|\overline{X} - \Exp{\overline{X}}\right| \geq t} \leq 
2 \exp\left(- \frac{2 n^2 t^2}{\sum_{j=1}^n (b_j-a_j)^2} \right),\ 
\hbox{where $\overline{X} = \frac1n \sum_{j=1}^n X_j$.}\end{equation*}
We will use $X_j = \Im\left[
\frac{ e^{2 \pi
i \Delta_d(v_j)/{\lambda}} }{|\!|q-v_j|\!|_2}\right]
 \in [-\frac{1}{d-r},\frac{1}{d-r}] =: [a_j,b_j]$ denoting the signal produced by each of the $m$ non central senders. Furthermore, we set $t = \frac12 \Exp{\overline{X}}$. 
 Hence, $a_j-b_j=\frac{2}{d-r} \geq \frac{2}{d-\frac1{15} d}\geq \frac{30}{14 d} $.

 Note that $\Exp{\overline{X}} \geq c_{6} \sqrt{\frac{\lambda}{r}} \frac{1}d  \geq c_{6}\sqrt{\frac{\lambda}{r}} \frac{14}{30} (a-b) $,  where $c_{6}=\frac{9}{4,480\pi\sqrt2}$.
  Therefore 
 $$\frac{t^2}{(b-a)^2} = \frac{\frac14\Exp{\overline{X}}^2 }{(b-a)^2} \geq 
 c_{6}^2 \left(\frac{14}{30}\right)^2 \frac{\lambda}{r} $$
 This implies the following Lemma.

 \begin{lemma}For $d\geq 15r$ and  $c_4=\frac12 c_{6}$ and $c_{7} = 2c_{6}^2 \left(\frac{14}{30}\right)^2$
$$ \Prob{\Im[\hbox{\sl rx}] \leq c_4\frac{m\sqrt{\lambda}}{d\sqrt{r}}} \leq 
2 \exp\left(-  c_{7} \frac{\lambda m}{r}  \right)\ .$$
\end{lemma}

 This implies the following Lemma.

\begin{lemma}\label{lakjfd}
For $r \geq r_1$ and $d>15 r$:
\begin{equation}
\Prob{|rx|^2 \leq  \left(c_4 \frac{m\sqrt{\lambda}}{d\sqrt{r}}\right)^2} \leq \frac{1}{n^{c_5}}\ . \label{rxright}
\end{equation}
\end{lemma}
\begin{proof}
Note that $\Im(z) \geq a$ implies $|z|^2 = \Im(z)^2 + \Re(z)^2 \geq a^2$. It remains to show that 
$ \exp\left(-  c_{7} \frac{\lambda m}{r}\right) \leq \frac{1}{n^{c_5}}$. Using that $m \geq \frac12 \rho \pi r^2$ from Lemma~\ref{densifity} (\ref{dlower}) with probability  $1-e^{-\frac18 \rho \pi r^2}$.
So, the error probability is bounded for $n\geq 2 $ as $\frac{1}{n^{c_{5}}}$, where  $c_5 =\frac1{2 \ln 2} c_{2} c_{3}  c_{7} \pi  + 1$.
Using  $r\geq r_1 = c_{2}/\lambda$, $\rho \geq \frac{c_{3}}{\ln 2}\ln n$ we bound the error probability of Lemma~\ref{lakjfd} as follows.
\begin{eqnarray*} 
2 e^{-  c_{7}  \frac{\lambda m}{r} }
& \leq &  e^{ - \frac12 c_{7} \lambda   \rho \pi r + \ln 2}\\
& \leq &  e^{ - \frac12 c_{2} c_{7}   \rho \pi  + \ln 2}\\
& \leq &  e^{ - \frac1{2 \ln 2} c_{2} c_{3} c_{7}  \pi  \ln n   + \ln 2}\\
& \leq &  n^{ - \frac1{2 \ln 2} c_{2} c_{3} c_{7} \pi  + 1}\\
& \leq &  n^{ - c_5}\\
\end{eqnarray*}
\end{proof}
The right side inside the probability of (\ref{rxright}) is larger than the SNR threshold $\beta=1$ if
$d   \leq c_4 m r^{1/2}\lambda^{1/2} $. 
 This implies that all nodes in distance $d \leq  \frac14 c_4 \rho  r_{j}^{3/2}\lambda^{1/2} =c_1 \rho  r_{j}^{3/2}\lambda^{1/2} $ for $c_{1}=\frac14c_{4}$ can be informed in round $j$ with high probability. This completes the proof of Theorem~\ref{safety}.
%
\end{proof}
\begin{theorem} For constant wavelength $\lambda$
MISO broadcasting takes $ \mathcal{O}(\log \log n - \log \log \rho)$ rounds to broadcast the signal. 
\end{theorem}
\begin{proof}
The algorithm works in two phases. In the first phase, we inform all nodes in radius $r_1=\mathcal{O}(1)$ using the UDG Broadcast algorithm with single senders. 
This takes at most $\mathcal{O}(1)$ rounds.

In the second phase we use the phase shift $\varphi_\ell = |\!|v_\ell-v_0|\!|_2$ for all senders $v_\ell$. Now the radii increase double exponentially with $r_{j+1} = c_1 \rho r_j^{3/2} \lambda^{\frac12}$.  Note that 
$r_{j+1} \geq 15 r_j$ if 
$  r_j \geq \frac{225}{c_1^2 \rho^2\lambda} \geq  \frac{225}{c_1^2 c_3^2 \lambda\log^2 n}\geq \frac{c_{2}}{\lambda} = r_1$  which holds for large enough number of nodes  $n$.

After   $j= \mathcal{O}(\log \log n - \log \log \rho)$ rounds we have reached 
%
   \begin{eqnarray*}
r_j &=& r_1^{(\frac32)^{j}} \left(c_1\rho\lambda^{\frac12}\right)^{1+\frac32 +
\ldots + (\frac32)^{j-1}} \\
    &=&  r_t^{(\frac32)^{j}} \left(c_1\rho\lambda^{\frac12}\right)^{2(\frac32)^{j}-2}\\
  &\geq & R  = \sqrt{\frac{n}{\pi \rho}}\, .
   \end{eqnarray*}
Note that in the first round $15 r_1$ nodes are already informed. So, all nodes in distance $r_2$ can be informed and the minimum distance of $15 r_1$ from the senders is kept. In every next round we decrease $r_j$ by a factor of $1/15$ to ensure that this minimum distance of $15 r_j$. The above recursion changes only to the extent that $c_1$ is replaced by $c'=c_1/15$ without changing the asymptotic behavior.
\end{proof}
Because of the super-exponential growth of $r_i$ the same observation as in Corollary~\ref{c:snrlight} can be made for the MISO broadcast.
\begin{corollary} \label{c:misolight}
The speed of the MIMO Broadcast approaches the speed of light for growing $n$.\end{corollary}

\section{Conclusions}
We have compared the number of rounds of collaborative broadcasting in three communication models. All of them are derived from the far-field superposition MISO/MIMO model where the signal-to-noise ratio allows a communication range of one unit. 

The first is the Unit-Disk-Graph model. Typically, parallel communication is seen here as a problem, resulting in the Radio Broadcasting model. For the Unit Disk Graph such interference results only in an extra overhead of a constant factor \cite{gandhi2008minimizing}. The delimiting factor is the diameter of the graph,  proportional to $\sqrt{n/\rho}$.


For the SNR-model one can achieve broadcasting in a logarithmic number of rounds which comes from the addition of the senders' signal energy. This allows to extend the disk of informed nodes by a factor of $\Theta(\sqrt{\rho})$, where $\rho$ is the sender density. This result is not surprising, since the area grows quadratically in the diameter and the path loss is to the power of two as well, which fits well to the law of conservation of energy.

For the MIMO model it is already known that beamforming increases the energy beyond the SNR model. It is possible to achieve logarithmic number of rounds for unicast on the line \cite{JS13_Beamforming_Line} and $\mathcal{O}(\log \log n)$ for the plane \cite{JS14_Beamforming_LogLog_TR}. The used beams are very narrow and for a growing number of sender nodes the ratio between beam range and angle decreases. So, the extended range results  from the focus of the energy on a smaller beam, while the signal energy is reduced elsewhere.

This leads to the question, how it is possible that a broadcast with high signal strength can take place. The answer is, that this panoramic beam draws its increased energy from a signal reduction from the three-dimensional space into the two-dimensional space, where the receivers happen to be located. 

\section{Outlook}
 
We have focused on broadcasting only a single sinusoidal signal and not a message consisting of many different signals. There, one faces inter-signal interference and inter-symbol interference. For inter-signal interference, note that the received signal has a constant phase shift. We bounded the interference of non-synchronized signals only by a constant factor with respect to the main signal. So, tighter bounds are necessary.
For the inter-symbol interference, a special encoding may reduce the interference caused by signals used for other parts of the message.

In \cite{JS14_Beamforming_LogLog_SSS} we have described a $\mathcal{O}(\log \log n)$ unicast algorithm without the need of perfect synchronization. There is some hope that a similar technique works here, too. Furthermore, the question remains open whether flooding in MISO is as efficient as the broadcasting algorithm relying on the Expanding Disk Algorithm. 

A drawback of the MISO broadcasting is that we assume the wavelength is constant, while the density grows logarithmic. So, the average distance of two nodes is smaller than the wavelength. Since, we use far-distant sending ranges this does not affect the validity of the claims. But a broadcasting time bound for constant density and very small wavelength seems desirable. For a constant density, we have to overcome the problem that the first sender might have no neighbors in the unit disk range. For this, on can focus only on situations where the broadcasting is possible in the UDG model. For small wavelength the solution seems to be more difficult. One approach might be to use the SNR broadcast for small ranges and then switch to MIMO broadcast. However, the SNR model equals the MISO/MIMO model for random phases only in the expectation and not with high probability. So, one has to analyze a SNR broadcast algorithm in the MISO/MIMO model where not all nodes might be triggered (a realistic problem known as Rayleigh fading). If this approach works out, then an additional number of $\mathcal{O}(\log (1/\lambda))$ steps are needed in MIMO broadcast, which is asymptotically better than the SNR broadcast if $1/\lambda$ is smaller than any polynomial.

Another interesting question concerns the influence of the path loss exponent $\alpha$, which we choose as $\alpha=2$. It has no influence to the UDG model. As an anonymous reviewer pointed out in the SNR model one expects for $\alpha<2$ a bound of $\mathcal{O}(\log \log n)$ for broadcasting, for $\alpha=2$ we have proved a bound of $\Theta(\log n)$ and for $\alpha>2$ a bound of $\mathcal{O}(n^{1/2})$ like in the UDG model is to be expected.
 
We conjecture for MIMO that our results can be generalized for $\alpha < 3$ because of area size of around $\Theta(r^{3/2} \lambda^{1/2})$ of nearly synchronous senders. For larger path loss the asymptotic number of rounds increase. For $\alpha=3$ we expect a logarithmic bound and for $\alpha>3$ the same behavior as in the Unit Disk Graph. 

Finally, the communication model is still very simple. Instead of a constant signal-to-noise ratio one might consider Gaussian noise. It is also unclear how obstacles influence the algorithm, or for which other path loss exponents, the double logarithmic number of rounds can be guaranteed.

\section{Acknowledgments}
We like to thank the organizers of the Dagstuhl Seminar 17271, July 2 - 7, 2017, Foundations of Wireless Networking, where this research has begun and first results have been found. 
We would like to thank Alexander Leibold, who performed and checked the automated proofs and anonymous reviewers of a previous version for their detailed and valuable input. We would also like to thank Tigrun Tonoyan, Magnus Halldorrson and Zvi Lotker for many fruitful discussions.

%
\bibliographystyle{unsrt}


\begin{thebibliography}{10}

\bibitem{peleg2007time}
David Peleg.
\newblock Time-efficient broadcasting in radio networks: A review.
\newblock In {\em International Conference on Distributed Computing and
  Internet Technology}, pages 1--18. Springer, 2007.

\bibitem{UnitDiskGraphs1990}
Brent~N. Clark, Charles~J. Colbourn, and David~S. Johnson.
\newblock {Unit disk graphs}.
\newblock {\em Discrete Mathematics}, 86(1--3):165--177, 1990.

\bibitem{Goussevskaia07}
Olga Goussevskaia, Yvonne~Anne Oswald, and Rogert Wattenhofer.
\newblock {Complexity in Geometric SINR}.
\newblock In {\em Proceedings of the 8th ACM international symposium on Mobile
  ad hoc networking and computing}, MobiHoc '07, pages 100--109, New York, NY,
  USA, 2007. ACM.

\bibitem{Tse_fundamentals_book}
David Tse and Pramod Viswanath.
\newblock {\em Fundamentals of wireless communication}.
\newblock Cambridge University Press, New York, NY, USA, 2005.

\bibitem{gandhi2008minimizing}
Rajiv Gandhi, Arunesh Mishra, and Srinivasan Parthasarathy.
\newblock Minimizing broadcast latency and redundancy in ad hoc networks.
\newblock {\em IEEE/ACM Transactions on Networking (TON)}, 16(4):840--851,
  2008.

\bibitem{Halldorsson:2018:LIS:3212734.3212766}
Magn\'{u}s~M. Halld\'{o}rsson and Tigran Tonoyan.
\newblock Leveraging indirect signaling for topology inference and fast
  broadcast.
\newblock In {\em Proceedings of the 2018 ACM Symposium on Principles of
  Distributed Computing}, PODC '18, pages 85--93, New York, NY, USA, 2018. ACM.

\bibitem{UnitDiskSINRLotker2009}
E.~Lebhar and Z.~Lotker.
\newblock {Unit disk graph and physical interference model: Putting pieces
  together}.
\newblock In {\em IEEE International Symposium on Parallel Distributed
  Processing (IPDPS 2009)}, pages 1--8, may 2009.

\bibitem{5779066}
F.~Ferrari, M.~Zimmerling, L.~Thiele, and O.~Saukh.
\newblock Efficient network flooding and time synchronization with glossy.
\newblock In {\em Proceedings of the 10th ACM/IEEE International Conference on
  Information Processing in Sensor Networks}, pages 73--84, Chicago, IL, USA,
  April 2011. IEEE.

\bibitem{Sirkeci-Mergen_First}
Birsen Sirkeci-Mergen, Anna Scaglione, and G\"{o}khan Mergen.
\newblock Asymptotic analysis of multistage cooperative broadcast in wireless
  networks.
\newblock {\em IEEE/ACM Trans. Netw.}, 14(SI):2531--2550, June 2006.

\bibitem{xue2004number}
Feng Xue and Panganamala~R Kumar.
\newblock The number of neighbors needed for connectivity of wireless networks.
\newblock {\em Wireless networks}, 10(2):169--181, 2004.

\bibitem{Gupta00thecapacity}
Piyush Gupta and P.~R. Kumar.
\newblock {The Capacity of Wireless Networks}.
\newblock {\em IEEE Transactions on Information Theory}, 46:388--404, 2000.

\bibitem{AvienEmek12}
Chen Avin, Yuval Emek, Erez Kantor, Zvi Lotker, David Peleg, and Liam Roditty.
\newblock {SINR Diagrams: Convexity and Its Applications in Wireless Networks}.
\newblock {\em J. ACM}, 59(4):18, 2012.

\bibitem{dpp2000}
Lun Dong, AP. Petropulu, and H.V. Poor.
\newblock A cross-layer approach to collaborative beamforming for wireless ad
  hoc networks.
\newblock {\em Signal Processing, IEEE Transactions on}, 56(7):2981--2993, July
  2008.

\bibitem{freitas2012energyWSN}
Edison~Pignaton de~Freitas, Jo{\~a}o Paulo C~Lustosa da~Costa, Andr{\'e} Lima~F
  de~Almeida, and Marco Marinho.
\newblock Applying mimo techniques to minimize energy consumption for long
  distances communications in wireless sensor networks.
\newblock In {\em Internet of Things, Smart Spaces, and Next Generation
  Networking}, pages 379--390. Springer, 2012.

\bibitem{NGS09_Linear_Capacity_Beamforming}
Urs Niesen, Piyush Gupta, and Devavrat Shah.
\newblock {On Capacity Scaling in Arbitrary Wireless Networks}.
\newblock {\em IEEE Transactions on Information Theory}, 55(9):3959--3982,
  2009.

\bibitem{ozgur2010linearCapacity}
Ayfer \"Ozg\"ur, Olivier Leveque, and David Tse.
\newblock {Hierarchical Cooperation Achieves Optimal Capacity Scaling in Ad Hoc
  Networks}.
\newblock {\em IEEE Transactions on Information Theory}, 53(10):3549--3572,
  October 2007.

\bibitem{franceschetti2009capacity}
Massimo Franceschetti, Marco~Donald Migliore, and Paolo Minero.
\newblock The capacity of wireless networks: Information-theoretic and physical
  limits.
\newblock {\em IEEE Transactions on Information Theory}, 55(8):3413--3424,
  2009.

\bibitem{ozgur2013spatial}
Ayfer {\"O}zg{\"u}r, Olivier L{\'e}v{\^e}que, and David Tse.
\newblock Spatial degrees of freedom of large distributed mimo systems and
  wireless ad hoc networks.
\newblock {\em IEEE Journal on Selected Areas in Communications},
  31(EPFL-ARTICLE-185421):202--214, 2013.

\bibitem{oyman2007power}
Ozgur Oyman and Arogyaswami~J Paulraj.
\newblock Power-bandwidth tradeoff in dense multi-antenna relay networks.
\newblock {\em IEEE Transactions on Wireless Communications}, 6(6), 2007.

\bibitem{oyman2011cooperative}
Ozgur Oyman and Arogyaswami~J Paulraj.
\newblock Cooperative {OFDMA} and distributed {MIMO} relaying over dense
  wireless networks, September~27 2011.
\newblock US Patent 8,027,301.

\bibitem{mlo13_telescopic_beamforming}
Alla Merzakreeva, Ayfer \"Ozg\"ur, and Olivier L{\'e}v{\^e}que.
\newblock Telescopic beamforming for large wireless networks.
\newblock In {\em IEEE Int. Symposium on Information Theory}, Istanbul, 2013.

\bibitem{JS14_Beamforming_LogLog_TR}
Thomas Janson and Christian Schindelhauer.
\newblock {Ad-Hoc Network Unicast in O(log log n) using Beamforming}.
\newblock http://arxiv.org/abs/1405.0417, May 2014.

\bibitem{JS13_Beamforming_Line}
Thomas Janson and Christian Schindelhauer.
\newblock {Broadcasting in Logarithmic Time for Ad Hoc Network Nodes on a Line
  using MIMO}.
\newblock In {\em Proceedings of the 25th ACM Symposium on Parallelism in
  Algorithms and Architectures, SPAA'13}. ACM, July 2013.

\bibitem{diss-janson-2015}
Thomas Janson.
\newblock {\em Energy-Efficient Collaborative Beamforming in Wireless Ad Hoc
  Networks}.
\newblock PhD thesis, University of Freiburg, Germany, 2015.

\bibitem{6962163}
Thomas Janson and Christian Schindelhauer.
\newblock Cooperative beamforming in ad-hoc networks with sublinear
  transmission power.
\newblock In {\em IEEE 10th International Conference on Wireless and Mobile
  Computing, Networking and Communications (WiMob)}, pages 144--151, Larnaca,
  Cyprus, Oct 2014. IEEE.

\bibitem{Ferrari:2012:LWB:2426656.2426658}
Federico Ferrari, Marco Zimmerling, Luca Mottola, and Lothar Thiele.
\newblock Low-power wireless bus.
\newblock In {\em Proceedings of the 10th ACM Conference on Embedded Network
  Sensor Systems}, SenSys '12, pages 1--14, New York, NY, USA, 2012. ACM.

\bibitem{sutton2015zippy}
Felix Sutton, Bernhard Buchli, Jan Beutel, and Lothar Thiele.
\newblock Zippy: On-demand network flooding.
\newblock In {\em Proceedings of the 13th ACM Conference on Embedded Networked
  Sensor Systems}, pages 45--58. ACM, 2015.

\bibitem{kumberg2017exploiting}
Timo Kumberg, Christian Schindelhauer, and Leonhard Reindl.
\newblock Exploiting concurrent wake-up transmissions using beat frequencies.
\newblock {\em Sensors}, 17(8):1717, 2017.

\bibitem{sirkeci2010broadcast}
Birsen Sirkeci-Mergen and Michael~C Gastpar.
\newblock On the broadcast capacity of wireless networks with cooperative
  relays.
\newblock {\em IEEE Transactions on Information Theory}, 56(8):3847--3861,
  2010.

\bibitem{jeon2007two}
Sang-Woon Jeon and Sae-Young Chung.
\newblock Two-phase opportunistic broadcasting in large wireless networks.
\newblock In {\em Information Theory, 2007. ISIT 2007. IEEE International
  Symposium on}, pages 2771--2775. IEEE, 2007.

\bibitem{2012RandomMIMOJanson}
Thomas Janson and Christian Schindelhauer.
\newblock {Analyzing Randomly Placed Multiple Antennas for MIMO Wireless
  Communication}.
\newblock In {\em {Fifth International Workshop on Selected Topics in Mobile
  and Wireless Computing (IEEE STWiMob)}}, Barcelona, 2012.

\bibitem{JS14_Beamforming_LogLog_SSS}
Thomas Janson and Christian Schindelhauer.
\newblock {Self-Synchronized Cooperative Beamforming in Ad-Hoc Networks}.
\newblock In {\em 16th International Symposium on Stabilization, Safety, and
  Security of Distributed Systems (SSS'14)}, Paderborn, Germany, September
  2014.

\bibitem{ozgur2007hierarchical}
Ayfer Ozgur, Olivier L{\'e}v{\^e}que, and NC~David.
\newblock Hierarchical cooperation achieves optimal capacity scaling in ad hoc
  networks.
\newblock {\em IEEE Transactions on information theory}, 53(10):3549--3572,
  2007.

\bibitem{oakmasterthesis}
Aditya Oak.
\newblock Analysis of a collaborative iterative miso broadcasting algorithm.
\newblock Master's thesis, University of Freiburg, George-Koehler-Allee 101,
  79110 Freiburg, Germany, March 2018.

\bibitem{hickey2001interval}
Timothy Hickey, Qun Ju, and Maarten~H Van~Emden.
\newblock Interval arithmetic: From principles to implementation.
\newblock {\em Journal of the ACM (JACM)}, 48(5):1038--1068, 2001.

\bibitem{hoeffding1963probability}
Wassily Hoeffding.
\newblock Probability inequalities for sums of bounded random variables.
\newblock {\em Journal of the American statistical association},
  58(301):13--30, 1963.

\end{thebibliography}
\end{document}